\documentclass[10pt,journal,compsoc]{IEEEtran}
\PassOptionsToPackage{hyphens}{url}
\usepackage[utf8]{inputenc}
\usepackage{amsmath} 
\usepackage{amssymb} 
\usepackage{bussproofs} 
\usepackage{rotating} 
\usepackage{hyperref}
\usepackage[hyphenbreaks]{breakurl} 
\usepackage[nocompress]{cite}

\usepackage{amsthm}
\newtheorem{definition}{Definition}
\newtheorem{theorem}{Theorem}
\newtheorem{lemma}[theorem]{Lemma}
\newtheorem*{convergence-thm}{Theorem}

\usepackage{tikz}
\usetikzlibrary{arrows}

\newif\ifincludeappendix
\includeappendixtrue

\newcommand{\evalto}{\;\Longrightarrow\;}

\newcommand{\placeholder}{%
  \makebox[0.7em]{%
    \kern.07em
    \vrule height.3ex
    \hrulefill
    \vrule height.3ex
    \kern.07em
  }%
}

\begin{document}
\sloppy
\title{A Conflict-Free Replicated JSON Datatype}
\author{Martin Kleppmann and Alastair R. Beresford
\thanks{M. Kleppmann and A.R. Beresford are with the University of Cambridge Computer Laboratory, Cambridge, UK.\protect\\Email: \url{mk428@cl.cam.ac.uk}, \url{arb33@cl.cam.ac.uk}.}}

\IEEEtitleabstractindextext{%
\begin{abstract}
Many applications model their data in a general-purpose storage format such as JSON. This data structure is modified by the application as a result of user input. Such modifications are well understood if performed sequentially on a single copy of the data, but if the data is replicated and modified concurrently on multiple devices, it is unclear what the semantics should be. In this paper we present an algorithm and formal semantics for a JSON data structure that automatically resolves concurrent modifications such that no updates are lost, and such that all replicas converge towards the same state (a conflict-free replicated datatype or CRDT). It supports arbitrarily nested list and map types, which can be modified by insertion, deletion and assignment. The algorithm performs all merging client-side and does not depend on ordering guarantees from the network, making it suitable for deployment on mobile devices with poor network connectivity, in peer-to-peer networks, and in messaging systems with end-to-end encryption.
\end{abstract}

}
\maketitle

\IEEEraisesectionheading{\section{Introduction}\label{sec:introduction}}

\IEEEPARstart{U}{sers} of mobile devices, such as smartphones, expect applications to continue working while the device is offline or has poor network connectivity, and to synchronize its state with the user's other devices when the network is available. Examples of such applications include calendars, address books, note-taking tools, to-do lists, and password managers. Similarly, collaborative work often requires several people to simultaneously edit the same text document, spreadsheet, presentation, graphic, and other kinds of document, with each person's edits reflected on the other collaborators' copies of the document with minimal delay.

What these applications have in common is that the application state needs to be replicated to several devices, each of which may modify the state locally. The traditional approach to concurrency control, serializability, would cause the application to become unusable at times of poor network connectivity~\cite{Davidson:1985hv}. If we require that applications work regardless of network availability, we must assume that users can make arbitrary modifications concurrently on different devices, and that any resulting conflicts must be resolved.

The simplest way to resolve conflicts is to discard some modifications when a conflict occurs, for example using a ``last writer wins'' policy. However, this approach is undesirable as it incurs data loss. An alternative is to let the user manually resolve the conflict, which is tedious and error-prone, and therefore should be avoided whenever possible.

Current applications solve this problem with a range of ad-hoc and application-specific mechanisms. In this paper we present a general-purpose datatype that provides the full expressiveness of the JSON data model, and supports concurrent modifications without loss of information. As we shall see later, our approach typically supports the automatic merging of concurrent modifications into a JSON data structure. We introduce a single, general mechanism (a multi-value register) into our model to record conflicting updates to leaf nodes in the JSON data structure. This mechanism then provides a consistent basis on which applications can resolve any remaining conflicts through programmatic means, or via further user input.  We expect that implementations of this datatype will drastically simplify the development of collaborative and state-synchronizing applications for mobile devices.

\subsection{JSON Data Model}

JSON is a popular general-purpose data encoding format, used in many databases and web services. It has similarities to XML, and we compare them in Section~\ref{sec:json-xml}. The structure of a JSON document can optionally be constrained by a schema; however, for simplicity, this paper discusses only untyped JSON without an explicit schema.

A JSON document is a tree containing two types of branch node:

\begin{description}
\item[Map:] A node whose children have no defined order, and where each child is labelled with a string \emph{key}. A key uniquely identifies one of the children. We treat keys as immutable, but values as mutable, and key-value mappings can be added and removed from the map. A JSON map is also known as an \emph{object}.
\item[List:] A node whose children have an order defined by the application. The list can be mutated by inserting or deleting list elements. A JSON list is also known as an \emph{array}.
\end{description}

A child of a branch node can be either another branch node, or a leaf node. A leaf of the tree contains a primitive value (string, number, boolean, or null). We treat primitive values as immutable, but allow the value of a leaf node to be modified by treating it as a \emph{register} that can be assigned a new value.

This model is sufficient to express the state of a wide range of applications. For example, a text document can be represented by a list of single-character strings; character-by-character edits are then expressed as insertions and deletions of list elements. In Section~\ref{sec:examples} we describe four more complex examples of using JSON to model application data.

\subsection{Replication and Conflict Resolution}\label{sec:intro-replication}

We consider systems in which a full copy of the JSON document is replicated on several devices. Those devices could be servers in datacenters, but we focus on mobile devices such as smartphones and laptops, which have intermittent network connectivity. We do not distinguish between devices owned by the same user and different users. Our model allows each device to optimistically modify its local replica of the document, and to asynchronously propagate those edits to other replicas.

Our only requirement of the network is that messages sent by one replica are eventually delivered to all other replicas, by retrying if delivery fails. We assume the network may arbitrarily delay, reorder and duplicate messages.

Our algorithm works client-side and does not depend on any server to transform or process messages. This approach allows messages to be delivered via a peer-to-peer network as well as a secure messaging protocol with end-to-end encryption~\cite{Unger:2015kg}. The details of the network implementation and cryptographic protocols are outside of the scope of this paper.

In Section~\ref{sec:semantics} we define formal semantics describing how conflicts are resolved when a JSON document is concurrently modified on different devices. Our design is based on three simple principles:
\begin{enumerate}
\item All replicas of the data structure should automatically converge towards the same state (a requirement known as \emph{strong eventual consistency}~\cite{Shapiro:2011un}).
\item No user input should be lost due to concurrent modifications.
\item If all sequential permutations of a set of updates lead to the same state, then concurrent execution of those updates also leads to the same state~\cite{Bieniusa:2012gt}.
\end{enumerate}

\subsection{Our Contributions}

Our main contribution in this work is to define an algorithm and formal semantics for collaborative, concurrent editing of JSON data structures with automatic conflict resolution. Although similar algorithms have previously been defined for lists, maps and registers individually (see Section~\ref{sec:related}), to our knowledge this paper is the first to integrate all of these structures into an arbitrarily composable datatype that can be deployed on any network topology.

A key requirement of conflict resolution is that after any sequence of concurrent modifications, all replicas eventually converge towards the same state. In Section~\ref{sec:convergence} and the appendix we prove a theorem to show that our algorithm satisfies this requirement.

Composing maps and lists into arbitrarily nested structures opens up subtle challenges that do not arise in flat structures, due to the possibility of concurrent edits at different levels of the tree. We illustrate some of those challenges by example in Section~\ref{sec:examples}. Nested structures are an important requirement for many applications. Consequently, the long-term goal of our work is to simplify the development of applications that use optimistic replication by providing a general algorithm for conflict resolution whose details can largely be hidden inside an easy-to-use software library.

\section{Related Work}\label{sec:related}

In this section we discuss existing approaches to optimistic replication, collaborative editing and conflict resolution.

\subsection{Operational Transformation}\label{sec:related-ot}

Algorithms based on \emph{operational transformation} (OT) have long been used for collaborative editing applications~\cite{Ellis:1989ue,Ressel:1996wx,Sun:1998vf,Nichols:1995fd}. Most of them treat a document as a single ordered list (of characters, for example) and do not support nested tree structures that are required by many applications. Some algorithms generalize OT to editing XML documents~\cite{Davis:2002iv,Ignat:2003jy,Wang:2015vo}, which provides nesting of ordered lists, but these algorithms do not support key-value maps as defined in this paper (see Section~\ref{sec:json-xml}). The performance of OT algorithms degrades rapidly as the number of concurrent operations increases~\cite{Li:2006kd,Mehdi:2011ke}.

Most deployed OT collaboration systems, including Google Docs~\cite{DayRichter:2010tt}, Etherpad~\cite{Etherpad:2011um}, Novell Vibe~\cite{Spiewak:2010vw} and Apache Wave (formerly Google Wave~\cite{Wang:2015vo}), rely on a single server to decide on a total ordering of operations~\cite{Lemonik:2016wh}, a design decision inherited from the Jupiter system~\cite{Nichols:1995fd}. This approach has the advantage of making the transformation functions simpler and less error-prone~\cite{Imine:2003ks}, but it does not meet our requirements, since we want to support peer-to-peer collaboration without requiring a single server.

Many secure messaging protocols, which we plan to use for encrypted collaboration, do not guarantee that different recipients will see messages in the same order~\cite{Unger:2015kg}. Although it is possible to decide on a total ordering of operations without using a single server by using an atomic broadcast protocol~\cite{Defago:2004ji}, such protocols are equivalent to consensus~\cite{Chandra:1996cp}, so they can only safely make progress if a quorum of participants are reachable. We expect that in peer-to-peer systems of mobile devices participants will frequently be offline, and so any algorithm requiring atomic broadcast would struggle to reach a quorum and become unavailable. Without quorums, the strongest guarantee a system can give is causal ordering~\cite{Attiya:2015dm}.

The Google Realtime API~\cite{Google:2015vk} is to our knowledge the only implementation of OT that supports arbitrary nesting of lists and maps. Like Google Docs, it relies on a single server~\cite{Lemonik:2016wh}. As a proprietary product, details of its algorithms have not been published.

\subsection{CRDTs}\label{sec:related-crdts}

Conflict-free replicated datatypes (CRDTs) are a family of data structures that support concurrent modification and guarantee convergence of concurrent updates. They work by attaching additional metadata to the data structure, making modification operations commutative by construction. The JSON datatype described in this paper is a CRDT.

CRDTs for registers, counters, maps, and sets are well-known~\cite{Shapiro:2011un,Shapiro:2011wy}, and have been implemented in various deployed systems such as Riak~\cite{Brown:2014hs,Brown:2013wy}. For ordered lists, various algorithms have been proposed, including WOOT~\cite{Oster:2006wj}, RGA~\cite{Roh:2011dw}, Treedoc~\cite{Preguica:2009fz}, Logoot~\cite{Weiss:2010hx}, and LSEQ~\cite{Nedelec:2013ky}. Attiya et al.~\cite{Attiya:2016kh} analyze the metadata overhead of collaboratively edited lists, and provide a correctness proof of the RGA algorithm. However, none of them support nesting: all of the aforementioned algorithms assume that each of their elements is a primitive value, not another CRDT.

The problem of nesting one CRDT inside another (also known as \emph{composition} or \emph{embedding}) has only been studied more recently. Riak allows nesting of counters and registers inside maps, and of maps within other maps~\cite{Brown:2014hs,Brown:2013wy}. Embedding counters inside maps raises questions of semantics, which have been studied by Baquero, Almeida and Lerche~\cite{Baquero:2016iv}. Almeida et al.~\cite{Almeida:2016tk} also define delta mutations for nested maps, and Baquero et al.~\cite{Baquero:2015tm} define a theoretical framework for composition of state-based CRDTs, based on lattices. None of this work integrates CRDTs for ordered lists, but the treatment of causality in these datatypes forms a basis for the semantics developed in this paper.

Burckhardt et al.~\cite{Burckhardt:2012jy} define \emph{cloud types}, which are similar to CRDTs and can be composed. They define \emph{cloud arrays}, which behave similarly to our map datatype, and \emph{entities}, which are like unordered sets or relations; ordered lists are not defined in this framework.

On the other hand, Martin et al.~\cite{Martin:2010ih} generalize Logoot~\cite{Weiss:2010hx} to support collaborative editing of XML documents~-- that is, a tree of nested ordered lists without nested maps. As discussed in Section~\ref{sec:json-xml}, such a structure does not capture the expressiveness of JSON.

Although CRDTs for registers, maps and ordered lists have existed for years in isolation, we are not aware of any prior work that allows them all to be composed into an arbitrarily nested CRDT with a JSON-like structure.

\subsection{Other Approaches}\label{sec:related-other}

Many replicated data systems need to deal with the problem of concurrent, conflicting modifications, but the solutions are often ad-hoc. For example, in Dynamo~\cite{DeCandia:2007ui} and CouchDB, if several values are concurrently written to the same key, the database preserves all of these values, and leaves conflict resolution to application code -- in other words, the only datatype it supports is a multi-value register. Naively chosen merge functions often exhibit anomalies such as deleted items reappearing~\cite{DeCandia:2007ui}. We believe that conflict resolution is not a simple matter that can reasonably be left to application programmers.

Another frequently-used approach to conflict resolution is \emph{last writer wins} (LWW), which arbitrarily chooses one among several concurrent writes as ``winner'' and discards the others. LWW is used in Apache Cassandra, for example. It does not meet our requirements, since we want no user input to be lost due to concurrent modifications.

Resolving concurrent updates on tree structures has been studied in the context of file synchronization~\cite{Balasubramaniam:1998jh,Ramsey:2001ce}.

Finally, systems such as Bayou~\cite{Terry:1995dn} allow offline nodes to execute transactions tentatively, and confirm them when they are next online. This approach relies on all servers executing transactions in the same serial order, and deciding whether a transaction was successful depending on its preconditions. Bayou has the advantage of being able to express global invariants such as uniqueness constraints, which require serialization and cannot be expressed using CRDTs~\cite{Bailis:2014th}. Bayou's downside is that tentative transactions may be rolled back, requiring explicit handling by the application, whereas CRDTs are defined such that operations cannot fail after they have been performed on one replica.

\section{Composing Data Structures}\label{sec:composing}

In this section we informally introduce our approach to collaborative editing of JSON data structures, and illustrate some peculiarities of concurrent nested data structures. A formal presentation of the algorithm follows in Section~\ref{sec:semantics}.

\subsection{Concurrent Editing Examples}\label{sec:examples}

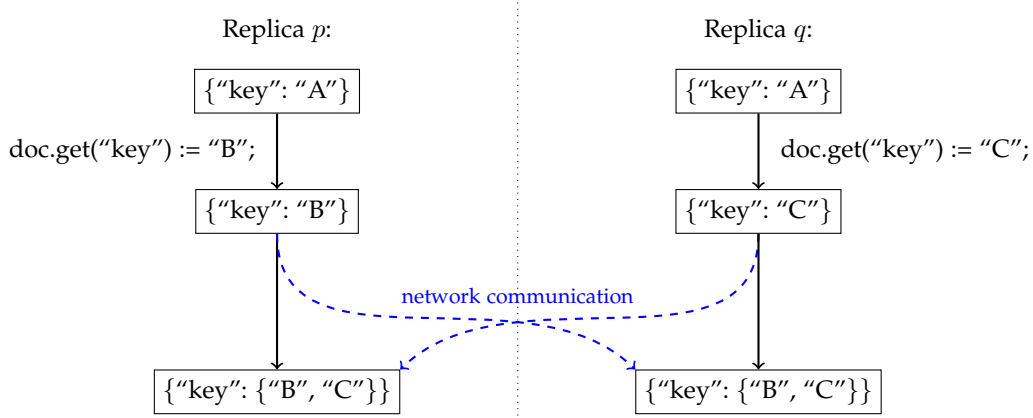
\begin{figure*}[p]
\centering
\begin{tikzpicture}[auto,scale=0.8]
\path [draw,dotted] (4,-0.5) -- (4,6.5);
\node (leftR)  at (0,6) {Replica $p$:};
\node (rightR) at (8,6) {Replica $q$:};
\node (left0)  at (0,5) [rectangle,draw] {\{``key'': ``A''\}};
\node (right0) at (8,5) [rectangle,draw] {\{``key'': ``A''\}};
\node (left1)  at (0,3) [rectangle,draw] {\{``key'': ``B''\}};
\node (right1) at (8,3) [rectangle,draw] {\{``key'': ``C''\}};
\node (left2)  at (0,0) [rectangle,draw] {\{``key'': \{``B'', ``C''\}\}};
\node (right2) at (8,0) [rectangle,draw] {\{``key'': \{``B'', ``C''\}\}};
\node (comms)  at (4,1.6) [text=blue] {\footnotesize network communication};
\draw [thick,->] (left0) to node [left,inner sep=8pt] {doc.get(``key'') := ``B'';} (left1);
\draw [thick,->] (right0) to node [right,inner sep=8pt] {doc.get(``key'') := ``C'';} (right1);
\draw [thick,->] (left1) -- (left2);
\draw [thick,dashed,blue,->] (left1.south)  to [out=270,in=135] (right2.north west);
\draw [thick,dashed,blue,->] (right1.south) to [out=270,in=45]  (left2.north east);
\draw [thick,->] (right1) -- (right2);
\end{tikzpicture}
\caption{Concurrent assignment to the register at doc.get(``key'') by replicas $p$ and $q$.}\label{fig:register-assign}
\end{figure*}

The sequential semantics of editing a JSON data structure are well-known, and the semantics of concurrently editing a flat map or list data structure have been thoroughly explored in the literature (see Section~\ref{sec:related}). However, when defining a CRDT for JSON data, difficulties arise due to the interactions between concurrency and nested data structures.

In the following examples we show some situations that might occur when JSON documents are concurrently modified, demonstrate how they are handled by our algorithm, and explain the rationale for our design decisions. In all examples we assume two replicas, labelled $p$ (drawn on the left-hand side) and $q$ (right-hand side). Local state for a replica is drawn in boxes, and modifications to local state are shown with labelled solid arrows; time runs down the page. Since replicas only mutate local state, we make communication of state changes between replicas explicit in our model. Network communication is depicted with dashed arrows.

Our first example is shown in Figure~\ref{fig:register-assign}. In a document that maps ``key'' to a register with value ``A'', replica $p$ sets the value of the register to ``B'', while replica $q$ concurrently sets it to ``C''. As the replicas subsequently exchange edits via network communication, they detect the conflict. Since we do not want to simply discard one of the edits, and the strings ``B'' and ``C'' cannot be meaningfully merged, the system must preserve both concurrent updates. This datatype is known as a \emph{multi-value register}: although a replica can only assign a single value to the register, reading the register may return a set of multiple values that were concurrently written.

A multi-value register is hardly an impressive CRDT, since it does not actually perform any conflict resolution. We use it only for primitive values for which no automatic merge function is defined. Other CRDTs could be substituted in its place: for example, a counter CRDT for a number that can only be incremented and decremented, or an ordered list of characters for a collaboratively editable string (see also Figure~\ref{fig:text-edit}).

\begin{figure*}[p]
\centering
\begin{tikzpicture}[auto,scale=0.8]
\path [draw,dotted] (4,-1) -- (4,8);
\node (left0)  at (0,7.5) [rectangle,draw] {\{``colors'': \{``blue'': ``\#0000ff''\}\}};
\node (right0) at (8,7.5) [rectangle,draw] {\{``colors'': \{``blue'': ``\#0000ff''\}\}};
\node [matrix] (left1) at (0,4) [rectangle,draw] {
    \node {\{``colors'': \{``blue'': ``\#0000ff'',}; \\
    \node {``red'': ``\#ff0000''\}\}}; \\
};
\node (right1) at (8,5.5) [rectangle,draw] {\{``colors'': \{\}\}};
\node (right2) at (8,3) [rectangle,draw] {\{``colors'': \{``green'': ``\#00ff00''\}\}};
\node [matrix] (left2) at (0,0) [rectangle,draw] {
    \node {\{``colors'': \{``red'': ``\#ff0000'',}; \\
    \node {``green'': ``\#00ff00''\}\}}; \\
};
\node [matrix] (right3) at (8,0) [rectangle,draw] {
    \node {\{``colors'': \{``red'': ``\#ff0000'',}; \\
    \node {``green'': ``\#00ff00''\}\}}; \\
};
\node (comms) at (4,2.1) [text=blue] {\footnotesize network communication};
\draw [thick,->] (left0)  -- (left1)
    node [left,text width=4.2cm,text centered,midway]  {doc.get(``colors'').get(``red'') := ``\#ff0000'';};
\draw [thick,->] (right0) to node [right] {doc.get(``colors'') := \{\};} (right1);
\draw [thick,->] (right1) -- (right2)
    node [right,text width=4.2cm,text centered,midway] {doc.get(``colors'').get(``green'') := ``\#00ff00'';};
\draw [thick,->] (left1)  to (left2);
\draw [thick,->] (right2) to (right3);
\draw [thick,dashed,blue,->] (left1.south)  to [out=270,in=135] (right3.north west);
\draw [thick,dashed,blue,->] (right2.south) to [out=270,in=45]  (left2.north east);
\end{tikzpicture}
\caption{Modifying the contents of a nested map while concurrently the entire map is overwritten.}\label{fig:map-remove}
\end{figure*}
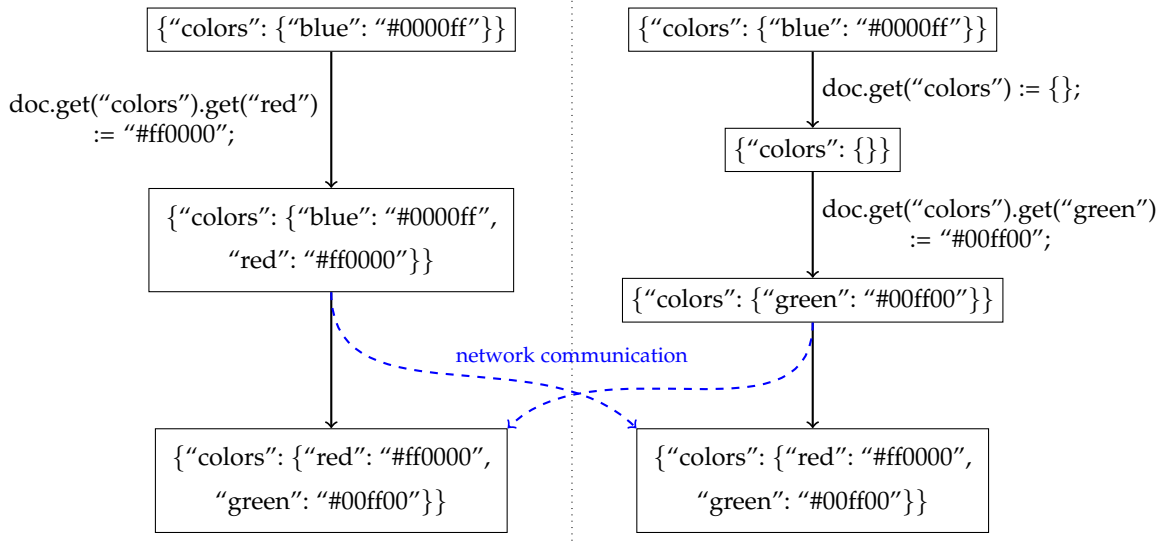

\begin{figure*}[p]
\centering
\begin{tikzpicture}[auto,scale=0.8]
\path [draw,dotted] (5,-0.5) -- (5,10);
\node (left0)  at  (0,9.5) [rectangle,draw] {\{\}};
\node (right0) at (10,9.5) [rectangle,draw] {\{\}};
\node (left1)  at  (0,7.5) [rectangle,draw] {\{``grocery'': []\}};
\node (right1) at (10,7.5) [rectangle,draw] {\{``grocery'': []\}};
\node (left2)  at  (0,5.0) [rectangle,draw] {\{``grocery'': [``eggs'']\}};
\node (left3)  at  (0,2.5) [rectangle,draw] {\{``grocery'': [``eggs'', ``ham'']\}};
\node (right2) at (10,5.0) [rectangle,draw] {\{``grocery'': [``milk'']\}};
\node (right3) at (10,2.5) [rectangle,draw] {\{``grocery'': [``milk'', ``flour'']\}};
\node (left4)  at  (0,0.0) [rectangle,draw] {\{``grocery'': [``eggs'', ``ham'', ``milk'', ``flour'']\}};
\node (right4) at (10,0.0) [rectangle,draw] {\{``grocery'': [``eggs'', ``ham'', ``milk'', ``flour'']\}};
\node (comms)  at  (5,1.4) [text=blue] {\footnotesize network communication};
\draw [thick,->] (left0)  to node [left]  {doc.get(``grocery'') := [];} (left1);
\draw [thick,->] (right0) to node [right] {doc.get(``grocery'') := [];} (right1);
\draw [thick,->] (left1)  -- (left2)
    node [left,text width=4cm,text centered,midway]  {doc.get(``grocery'').idx(0) .insertAfter(``eggs'');};
\draw [thick,->] (right1) -- (right2)
    node [right,text width=4cm,text centered,midway] {doc.get(``grocery'').idx(0) .insertAfter(``milk'');};
\draw [thick,->] (left2)  -- (left3)
    node [left,text width=4cm,text centered,midway]  {doc.get(``grocery'').idx(1) .insertAfter(``ham'');};
\draw [thick,->] (right2) -- (right3)
    node [right,text width=4cm,text centered,midway] {doc.get(``grocery'').idx(1) .insertAfter(``flour'');};
\draw [thick,->] (left3)  to (left4);
\draw [thick,->] (right3) to (right4);
\draw [thick,dashed,blue,->] (left3.south)  to [out=270,in=135] (right4.north west);
\draw [thick,dashed,blue,->] (right3.south) to [out=270,in=45]  (left4.north east);
\end{tikzpicture}
\caption{Two replicas concurrently create ordered lists under the same map key.}\label{fig:two-lists}
\end{figure*}
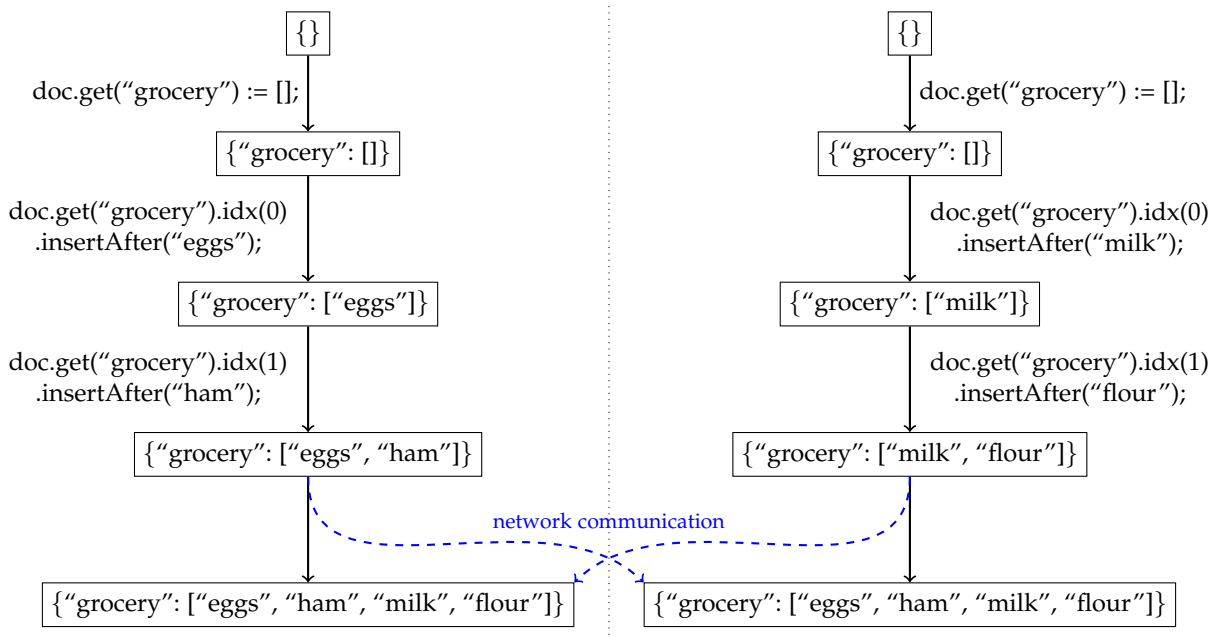

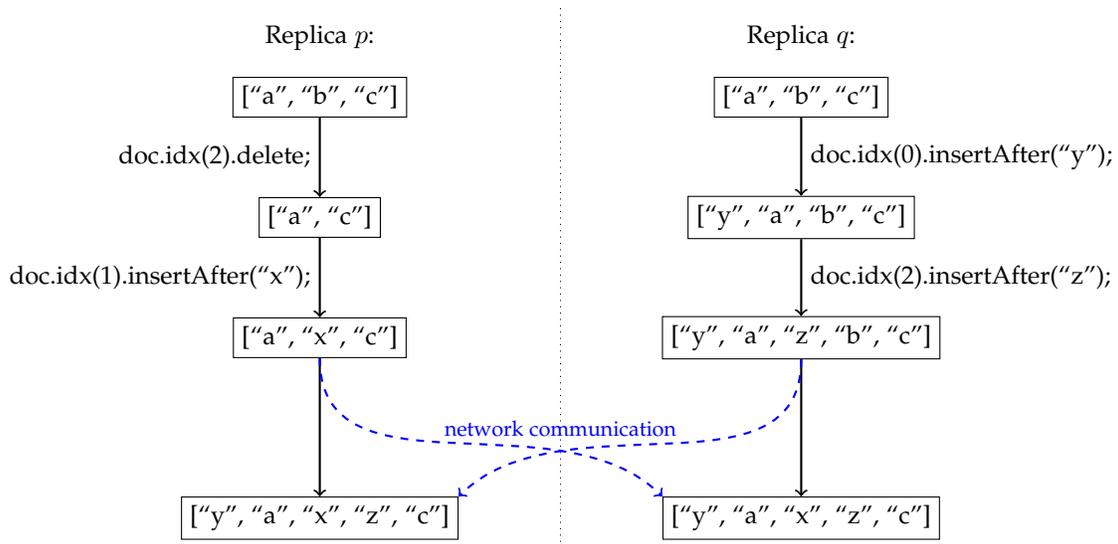
\begin{figure*}[p]
\centering
\begin{tikzpicture}[auto,scale=0.8]
\path [draw,dotted] (4,-0.5) -- (4,8.5);
\node (leftR)  at (0,8) {Replica $p$:};
\node (rightR) at (8,8) {Replica $q$:};
\node (left1)  at (0,7) [rectangle,draw] {[``a'', ``b'', ``c'']};
\node (left2)  at (0,5) [rectangle,draw] {[``a'', ``c'']};
\node (left3)  at (0,3) [rectangle,draw] {[``a'', ``x'', ``c'']};
\node (left4)  at (0,0) [rectangle,draw] {[``y'', ``a'', ``x'', ``z'', ``c'']};
\node (right1) at (8,7) [rectangle,draw] {[``a'', ``b'', ``c'']};
\node (right2) at (8,5) [rectangle,draw] {[``y'', ``a'', ``b'', ``c'']};
\node (right3) at (8,3) [rectangle,draw] {[``y'', ``a'', ``z'', ``b'', ``c'']};
\node (right4) at (8,0) [rectangle,draw] {[``y'', ``a'', ``x'', ``z'', ``c'']};
\node (comms)  at (4, 1.5) [text=blue] {\footnotesize network communication};
\draw [thick,->] (left1)  to node [left]  {doc.idx(2).delete;} (left2);
\draw [thick,->] (left2)  to node [left]  {doc.idx(1).insertAfter(``x'');} (left3);
\draw [thick,->] (right1) to node [right] {doc.idx(0).insertAfter(``y'');} (right2);
\draw [thick,->] (right2) to node [right] {doc.idx(2).insertAfter(``z'');} (right3);
\draw [thick,->] (left3)  to (left4);
\draw [thick,->] (right3) to (right4);
\draw [thick,dashed,blue,->] (left3.south)  to [out=270,in=135] (right4.north west);
\draw [thick,dashed,blue,->] (right3.south) to [out=270,in=45]  (left4.north east);
\end{tikzpicture}
\caption{Concurrent editing of an ordered list of characters (i.e., a text document).}\label{fig:text-edit}
\end{figure*}

\begin{figure*}[p]
\centering
\begin{tikzpicture}[auto,scale=0.8]
\path [draw,dotted] (4,-1) -- (4,7.5);
\node (left1)  at (0,7) [rectangle,draw] {\{\}};
\node (left2)  at (0,5) [rectangle,draw] {\{``a'': \{\}\}};
\node (left3)  at (0,3) [rectangle,draw] {\{``a'': \{``x'': ``y''\}\}};
\node [matrix] (left4) at (0,0) [rectangle,draw] {
    \node {\{mapT(``a''): \{``x'': ``y''\},}; \\
    \node {listT(``a''): [``z'']\}}; \\
};
\node (right1) at (8,7) [rectangle,draw] {\{\}};
\node (right2) at (8,5) [rectangle,draw] {\{``a'': []\}};
\node (right3) at (8,3) [rectangle,draw] {\{``a'': [``z'']\}};
\node [matrix] (right4) at (8,0) [rectangle,draw] {
    \node {\{mapT(``a''): \{``x'': ``y''\},}; \\
    \node {listT(``a''): [``z'']\}}; \\
};
\node (comms)  at (4,2.0) [text=blue] {\footnotesize network communication};
\draw [thick,->] (left1)  to node [left]  {doc.get(``a'') := \{\};} (left2);
\draw [thick,->] (left2)  to node [left]  {doc.get(``a'').get(``x'') := ``y'';} (left3);
\draw [thick,->] (right1) to node [right] {doc.get(``a'') := [];} (right2);
\draw [thick,->] (right2) to node [right] {doc.get(``a'').idx(0).insertAfter(``z'');} (right3);
\draw [thick,->] (left3)  to (left4);
\draw [thick,->] (right3) to (right4);
\draw [thick,dashed,blue,->] (left3.south)  to [out=270,in=135] (right4.north west);
\draw [thick,dashed,blue,->] (right3.south) to [out=270,in=45]  (left4.north east);
\end{tikzpicture}
\caption{Concurrently assigning values of different types to the same map key.}\label{fig:type-clash}
\end{figure*}
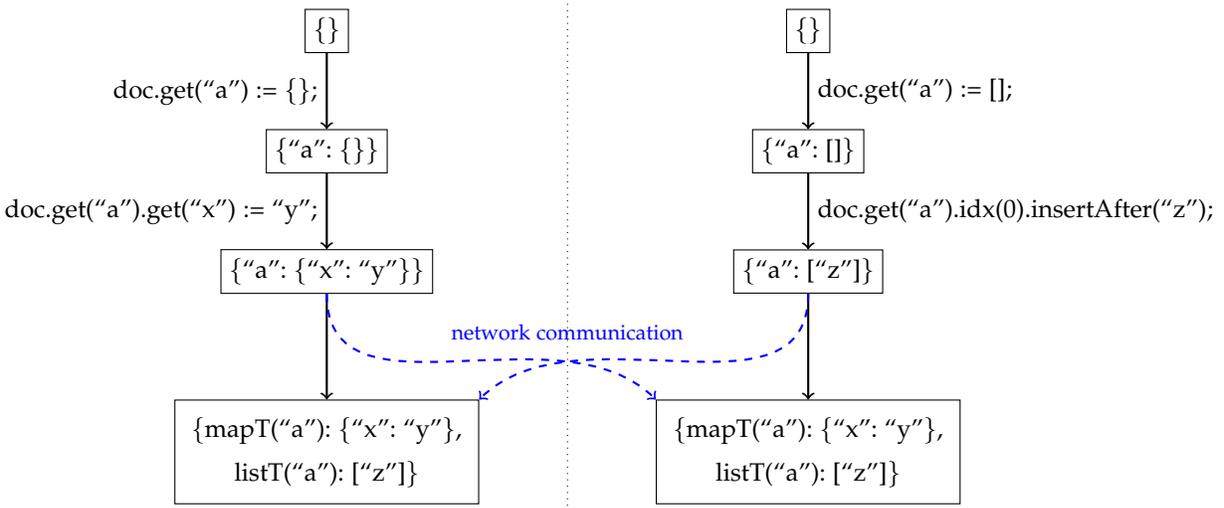

\begin{figure*}[p]
\centering
\begin{tikzpicture}[auto,scale=0.8]
\path [draw,dotted] (4,-0.5) -- (4,7.0);
\node [matrix] (left0) at (0,6) [rectangle,draw] {
    \node {\{``todo'': [\{``title'': ``buy milk'',}; \\
    \node {``done'': false\}]\}}; \\
};
\node [matrix] (right0) at (8,6) [rectangle,draw] {
    \node {\{``todo'': [\{``title'': ``buy milk'',}; \\
    \node {``done'': false\}]\}}; \\
};
\node (left1)  at (0,3) [rectangle,draw] {\{``todo'': []\}};
\node [matrix] (right1) at (8,3) [rectangle,draw] {
    \node {\{``todo'': [\{``title'': ``buy milk'',}; \\
    \node {``done'': true\}]\}}; \\
};
\node (left2)  at (0,0) [rectangle,draw] {\{``todo'': [\{``done'': true\}]\}};
\node (right2) at (8,0) [rectangle,draw] {\{``todo'': [\{``done'': true\}]\}};
\node (comms)  at (4,1.6) [text=blue] {\footnotesize network communication};
\draw [thick,->] (left0)  to node [left]  {doc.get(``todo'').idx(1).delete;} (left1);
\draw [thick,->] (right0) to node [right] {doc.get(``todo'').idx(1).get(``done'') := true;} (right1);
\draw [thick,->] (left1)  to (left2);
\draw [thick,->] (right1) to (right2);
\draw [thick,dashed,blue,->] (left1.south)  to [out=270,in=135] (right2.north west);
\draw [thick,dashed,blue,->] (right1.south) to [out=270,in=45]  (left2.north east);
\end{tikzpicture}
\caption{One replica removes a list element, while another concurrently updates its contents.}\label{fig:todo-item}
\end{figure*}
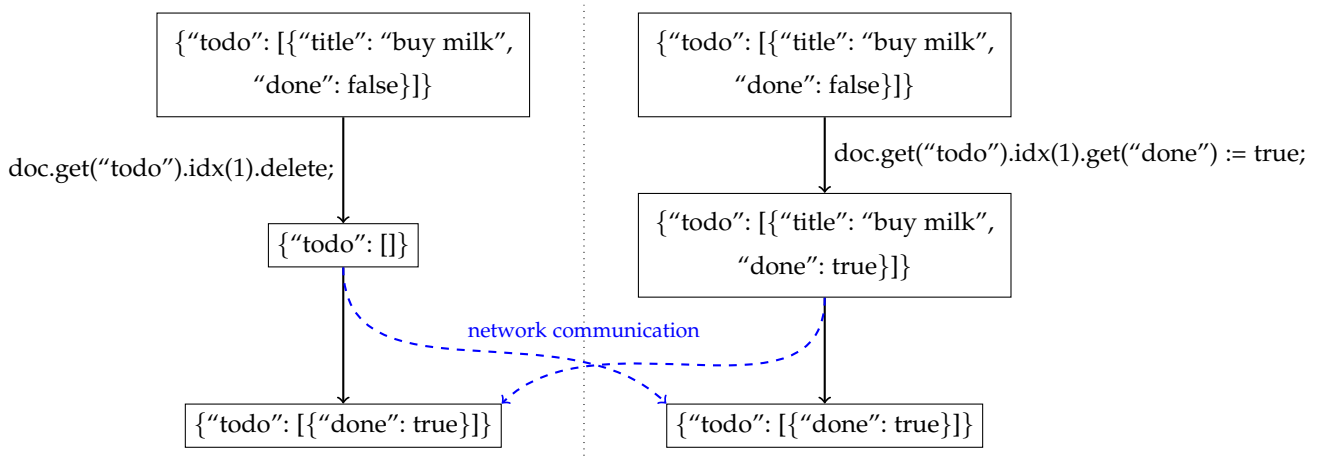

Figure~\ref{fig:map-remove} gives an example of concurrent edits at different levels of the JSON tree. Here, replica $p$ adds ``red'' to a map of colors, while replica $q$ concurrently blanks out the entire map of colors and then adds ``green''. Instead of assigning an empty map, $q$ could equivalently remove the entire key ``colors'' from the outer map, and then assign a new empty map to that key. The difficulty in this example is that the addition of ``red'' occurs at a lower level of the tree, while concurrently the removal of the map of colors occurs at a higher level of the tree.

One possible way of handling such a conflict would be to let edits at higher levels of the tree always override concurrent edits within that subtree. In this case, that would mean the addition of ``red'' would be discarded, since it would be overridden by the blanking-out of the entire map of colors. However, that behavior would violate our requirement that no user input should be lost due to concurrent modifications. Instead, we define merge semantics that preserve all changes, as shown in Figure~\ref{fig:map-remove}: ``blue'' must be absent from the final map, since it was removed by blanking out the map, while ``red'' and ``green'' must be present, since they were explicitly added. This behavior matches that of CRDT maps in Riak~\cite{Brown:2014hs,Brown:2013wy}.

Figure~\ref{fig:two-lists} illustrates another problem with maps: two replicas can concurrently insert the same map key. Here, $p$ and $q$ each independently create a new shopping list under the same map key ``grocery'', and add items to the list. In the case of Figure~\ref{fig:register-assign}, concurrent assignments to the same map key were left to be resolved by the application, but in Figure~\ref{fig:two-lists}, both values are lists and so they can be merged automatically. We preserve the ordering and adjacency of items inserted at each replica, so ``ham'' appears after ``eggs'', and ``flour'' appears after ``milk'' in the merged result. There is no information on which replica's items should appear first in the merged result, so the algorithm can make an arbitrary choice between ``eggs, ham, milk, flour'' and ``milk, flour, eggs, ham'', provided that all replicas end up with the items in the same order.

Figure~\ref{fig:text-edit} shows how a collaborative text editor can be implemented, by treating the document as a list of characters. All changes are preserved in the merged result: ``y'' is inserted before ``a''; ``x'' and ``z'' are inserted between ``a'' and ``c''; and ``b'' is deleted.

Figure~\ref{fig:type-clash} demonstrates a variant of the situation in Figure~\ref{fig:two-lists}, where two replicas concurrently insert the same map key, but they do so with different datatypes as values: $p$ inserts a nested map, whereas $q$ inserts a list. These datatypes cannot be meaningfully merged, so we preserve both values separately. We do this by tagging each map key with a type annotation (\textsf{mapT}, \textsf{listT}, or \textsf{regT} for a map, list, or register value respectively), so each type inhabits a separate namespace.

Finally, Figure~\ref{fig:todo-item} shows a limitation of the principle of preserving all user input. In a to-do list application, one replica removes a to-do item from the list, while another replica concurrently marks the same item as done. As the changes are merged, the update of the map key ``done'' effectively causes the list item to be resurrected on replica $p$, leaving a to-do item without a title (since the title was deleted as part of deleting the list item). This behavior is consistent with the example in Figure~\ref{fig:map-remove}, but it is perhaps surprising. In this case it may be more desirable to discard one of the concurrent updates, and thus preserve the implicit schema that a to-do item has both a ``title'' and a ``done'' field. We leave the analysis of developer expectations and the development of a schema language for future work.

\subsection{JSON Versus XML}\label{sec:json-xml}

The most common alternative to JSON is XML, and collaborative editing of XML documents has been previously studied~\cite{Davis:2002iv,Ignat:2003jy,Wang:2015vo}. Besides the superficial syntactical differences, the tree structure of XML and JSON appears quite similar. However, there is an important difference that we should highlight.

JSON has two collection constructs that can be arbitrarily nested: maps for unordered key-value pairs, and lists for ordered sequences. In XML, the children of an element form an ordered sequence, while the attributes of an element are unordered key-value pairs. However, XML does not allow nested elements inside attributes -- the value of an attribute can only be a primitive datatype. Thus, XML supports maps within lists, but not lists within maps. In this regard, XML is less expressive than JSON: the scenarios in Figures~\ref{fig:two-lists} and~\ref{fig:type-clash} cannot occur in XML.

Some applications may attach map-like semantics to the children of an XML document, for example by interpreting the child element name as key. However, this key-value structure is not part of XML itself, and would not be enforced by existing collaborative editing algorithms for XML. If multiple children with the same key are concurrently created, existing algorithms would create duplicate children with the same key rather than merging them like in Figure~\ref{fig:two-lists}.

\subsection{Document Editing API}\label{sec:editing-api}

\begin{figure}
\centering
\begin{tabular}{rcll}
CMD & ::= & \texttt{let} $x$ \texttt{=} EXPR & $x \in \mathrm{VAR}$ \\
& $|$ & $\mathrm{EXPR}$ \texttt{:=} $v$ & $v \in \mathrm{VAL}$ \\
& $|$ & $\mathrm{EXPR}$\texttt{.insertAfter(}$v$\texttt{)} & $v \in \mathrm{VAL}$ \\
& $|$ & $\mathrm{EXPR}$\texttt{.delete} \\
& $|$ & \texttt{yield} \\
& $|$ & CMD\texttt{;} CMD \vspace{0.5em}\\
EXPR & ::= & \texttt{doc} \\
& $|$ & $x$ & $x \in \mathrm{VAR}$ \\
& $|$ & EXPR\texttt{.get(}$\mathit{key}$\texttt{)} & $\mathit{key} \in \mathrm{String}$ \\
& $|$ & EXPR\texttt{.idx(}$i$\texttt{)} & $i \ge 0$ \\
& $|$ & EXPR\texttt{.keys} \\
& $|$ & EXPR\texttt{.values} \vspace{0.5em}\\
VAR & ::= & $x$ & $x \in \mathrm{VarString}$\\
VAL & ::= & $n$ & $n \in \mathrm{Number}$ \\
& $|$ & $\mathit{str}$ & $\mathit{str} \in \mathrm{String}$ \\
& $|$ & \texttt{true} $|$ \texttt{false} $|$ \texttt{null} \\
& $|$ & \verb|{}| $|$ \verb|[]|
\end{tabular}
\caption{Syntax of command language for querying and modifying a document.}\label{fig:local-syntax}
\end{figure}

\begin{figure}
\centering
\begin{verbatim}
doc := {};
doc.get("shopping") := [];
let head = doc.get("shopping").idx(0);
head.insertAfter("eggs");
let eggs = doc.get("shopping").idx(1);
head.insertAfter("cheese");
eggs.insertAfter("milk");

// Final state:
{"shopping": ["cheese", "eggs", "milk"]}
\end{verbatim}
\caption{Example of programmatically constructing a JSON document.}\label{fig:make-doc}
\end{figure}

To define the semantics for collaboratively editable data structures, we first define a simple command language that is executed locally at any of the replicas, and which allows that replica's local copy of the document to be queried and modified. Performing read-only queries has no side-effects, but modifying the document has the effect of producing \emph{operations} describing the mutation. Those operations are immediately applied to the local copy of the document, and also enqueued for asynchronous broadcasting to other replicas.

The syntax of the command language is given in Figure~\ref{fig:local-syntax}. It is not a full programming language, but rather an API through which the document state is queried and modified. We assume that the program accepts user input and issues a (possibly infinite) sequence of commands to the API. We model only the semantics of those commands, and do not assume anything about the program in which the command language is embedded. The API differs slightly from the JSON libraries found in many programming languages, in order to allow us to define consistent merge semantics.

We first explain the language informally, before giving its formal semantics. The expression construct EXPR is used to construct a \emph{cursor} which identifies a position in the document. An expression starts with either the special token \texttt{doc}, identifying the root of the JSON document tree, or a variable $x$ that was previously defined in a \texttt{let} command. The expression defines, left to right, the path the cursor takes as it navigates through the tree towards the leaves: the operator \texttt{.get(}$\mathit{key}$\texttt{)} selects a key within a map, and \texttt{.idx(}$n$\texttt{)} selects the $n$th element of an ordered list. Lists are indexed starting from 1, and \texttt{.idx(0)} is a special cursor indicating the head of a list (a virtual position before the first list element).

The expression construct EXPR can also query the state of the document: \texttt{keys} returns the set of keys in the map at the current cursor, and \texttt{values} returns the contents of the multi-value register at the current cursor. (\texttt{values} is not defined if the cursor refers to a map or list.)

A command CMD either sets the value of a local variable (\texttt{let}), performs network communication (\texttt{yield}), or modifies the document. A document can be modified by writing to a register (the operator \texttt{:=} assigns the value of the register identified by the cursor), by inserting an element into a list (\texttt{insertAfter} places a new element after the existing list element identified by the cursor, and \texttt{.idx(0).insertAfter} inserts at the head of a list), or by deleting an element from a list or a map (\texttt{delete} removes the element identified by the cursor).

Figure~\ref{fig:make-doc} shows an example sequence of commands that constructs a new document representing a shopping list. First \texttt{doc} is set to \verb|{}|, the empty map literal, and then the key \verb|"shopping"| within that map is set to the empty list \verb|[]|. The third line navigates to the key \verb|"shopping"| and selects the head of the list, placing the cursor in a variable called \texttt{head}. The list element ``eggs'' is inserted at the head of the list. In line 5, the variable \texttt{eggs} is set to a cursor pointing at the list element ``eggs''. Then two more list elements are inserted: ``cheese'' at the head, and ``milk'' after ``eggs''.

Note that the cursor \texttt{eggs} identifies the list element by identity, not by its index: after the insertion of ``cheese'', the element ``eggs'' moves from index 1 to 2, but ``milk'' is nevertheless inserted after ``eggs''. As we shall see later, this feature is helpful for achieving desirable semantics in the presence of concurrent modifications.

\begin{figure*}
\begin{center}
\AxiomC{$\mathit{cmd}_1 \mathbin{:} \mathrm{CMD}$}
\AxiomC{$A_p,\, \mathit{cmd}_1 \evalto A_p'$}
\LeftLabel{\textsc{Exec}}
\BinaryInfC{$A_p,\, \langle \mathit{cmd}_1 \mathbin{;} \mathit{cmd}_2 \mathbin{;} \dots \rangle
    \evalto A_p',\, \langle \mathit{cmd}_2 \mathbin{;} \dots \rangle$}
\DisplayProof\hspace{4em}
\AxiomC{}
\LeftLabel{\textsc{Doc}}
\UnaryInfC{$A_p,\, \mathsf{doc} \evalto \mathsf{cursor}(\langle\rangle,\, \mathsf{doc})$}
\DisplayProof\proofSkipAmount
\end{center}

\begin{center}
\AxiomC{$A_p,\, \mathit{expr} \evalto \mathit{cur}$}
\LeftLabel{\textsc{Let}}
\UnaryInfC{$A_p,\, \mathsf{let}\; x = \mathit{expr} \evalto A_p[\,x \,\mapsto\, \mathit{cur}\,]$}
\DisplayProof\hspace{3em}
\AxiomC{$x \in \mathrm{dom}(A_p)$}
\LeftLabel{\textsc{Var}}
\UnaryInfC{$A_p,\, x \evalto A_p(x)$}
\DisplayProof\proofSkipAmount
\end{center}

\begin{prooftree}
\AxiomC{$A_p,\, \mathit{expr} \evalto \mathsf{cursor}(\langle k_1, \dots, k_{n-1} \rangle,\, k_n)$}
\AxiomC{$k_n \not= \mathsf{head}$}
\LeftLabel{\textsc{Get}}
\BinaryInfC{$A_p,\, \mathit{expr}.\mathsf{get}(\mathit{key}) \evalto
    \mathsf{cursor}(\langle k_1, \dots, k_{n-1}, \mathsf{mapT}(k_n) \rangle,\, \mathit{key})$}
\end{prooftree}

\begin{prooftree}
\AxiomC{$A_p,\, \mathit{expr} \evalto \mathsf{cursor}(\langle k_1, \dots, k_{n-1} \rangle,\, k_n)$}
\AxiomC{$A_p,\, \mathsf{cursor}(\langle k_1, \dots, k_{n-1}, \mathsf{listT}(k_n) \rangle,\,
    \mathsf{head}).\mathsf{idx}(i) \evalto \mathit{cur}'$}
\LeftLabel{$\textsc{Idx}_1$}
\BinaryInfC{$A_p,\, \mathit{expr}.\mathsf{idx}(i) \evalto \mathit{cur}'$}
\end{prooftree}

\begin{prooftree}
\AxiomC{$k_1 \in \mathrm{dom}(\mathit{ctx})$}
\AxiomC{$\mathit{ctx}(k_1),\, \mathsf{cursor}(\langle k_2, \dots, k_{n-1} \rangle,\, k_n).\mathsf{idx}(i)
    \evalto \mathsf{cursor}(\langle k_2, \dots, k_{n-1} \rangle,\, k_n')$}
\LeftLabel{$\textsc{Idx}_2$}
\BinaryInfC{$\mathit{ctx},\, \mathsf{cursor}(\langle k_1, k_2, \dots, k_{n-1} \rangle,\, k_n).\mathsf{idx}(i)
    \evalto \mathsf{cursor}(\langle k_1, k_2, \dots, k_{n-1} \rangle,\, k_n')$}
\end{prooftree}

\begin{prooftree}
\AxiomC{$i>0 \,\wedge\, \mathit{ctx}(\mathsf{next}(k)) = k' \,\wedge\, k' \not= \mathsf{tail}$}
\AxiomC{$\mathit{ctx}(\mathsf{pres}(k')) \not= \{\}$}
\AxiomC{$\mathit{ctx},\, \mathsf{cursor}(\langle\rangle,\, k').\mathsf{idx}(i-1) \evalto \mathit{ctx}'$}
\LeftLabel{$\textsc{Idx}_3$}
\TrinaryInfC{$\mathit{ctx},\, \mathsf{cursor}(\langle\rangle,\, k).\mathsf{idx}(i) \evalto \mathit{ctx}'$}
\end{prooftree}

\begin{prooftree}
\AxiomC{$i>0 \,\wedge\, \mathit{ctx}(\mathsf{next}(k)) = k' \,\wedge\, k' \not= \mathsf{tail}$}
\AxiomC{$\mathit{ctx}(\mathsf{pres}(k')) = \{\}$}
\AxiomC{$\mathit{ctx},\, \mathsf{cursor}(\langle\rangle,\, k').\mathsf{idx}(i) \evalto \mathit{cur}'$}
\LeftLabel{$\textsc{Idx}_4$}
\TrinaryInfC{$\mathit{ctx},\, \mathsf{cursor}(\langle\rangle,\, k).\mathsf{idx}(i) \evalto \mathit{cur}'$}
\end{prooftree}

\begin{prooftree}
\AxiomC{$i=0$}
\LeftLabel{$\textsc{Idx}_5$}
\UnaryInfC{$\mathit{ctx},\, \mathsf{cursor}(\langle\rangle,\, k).\mathsf{idx}(i) \evalto
    \mathsf{cursor}(\langle\rangle,\, k)$}
\end{prooftree}

\[ \mathrm{keys}(\mathit{ctx}) = \{\; k \mid
    \mathsf{mapT}(k)  \in \mathrm{dom}(\mathit{ctx}) \,\vee\,
    \mathsf{listT}(k) \in \mathrm{dom}(\mathit{ctx}) \,\vee\,
    \mathsf{regT}(k)  \in \mathrm{dom}(\mathit{ctx})
\;\} \]

\begin{prooftree}
\AxiomC{$A_p,\, \mathit{expr} \evalto \mathit{cur}$}
\AxiomC{$A_p,\, \mathit{cur}.\mathsf{keys} \evalto \mathit{keys}$}
\LeftLabel{$\textsc{Keys}_1$}
\BinaryInfC{$A_p,\, \mathit{expr}.\mathsf{keys} \evalto \mathit{keys}$}
\end{prooftree}

\begin{prooftree}
\AxiomC{$\mathit{map} = \mathit{ctx}(\mathsf{mapT}(k))$}
\AxiomC{$\mathit{keys} = \{\; k \mid k \in \mathrm{keys}(\mathit{map}) \,\wedge\,
    \mathit{map}(\mathsf{pres}(k)) \not= \{\} \;\}$}
\LeftLabel{$\textsc{Keys}_2$}
\BinaryInfC{$A_p,\, \mathsf{cursor}(\langle\rangle,\, k).\mathsf{keys} \evalto \mathit{keys}$}
\end{prooftree}

\begin{prooftree}
\AxiomC{$k_1 \in \mathrm{dom}(\mathit{ctx})$}
\AxiomC{$\mathit{ctx}(k_1),\, \mathsf{cursor}(\langle k_2, \dots, k_{n-1} \rangle,\, k_n).\mathsf{keys}
    \evalto \mathit{keys}$}
\LeftLabel{$\textsc{Keys}_3$}
\BinaryInfC{$\mathit{ctx},\, \mathsf{cursor}(\langle k_1, k_2, \dots, k_{n-1} \rangle,\, k_n).\mathsf{keys}
    \evalto \mathit{keys}$}
\end{prooftree}

\begin{prooftree}
\AxiomC{$A_p,\, \mathit{expr} \evalto \mathit{cur}$}
\AxiomC{$A_p,\, \mathit{cur}.\mathsf{values} \evalto \mathit{val}$}
\LeftLabel{$\textsc{Val}_1$}
\BinaryInfC{$A_p,\, \mathit{expr}.\mathsf{values} \evalto \mathit{val}$}
\end{prooftree}

\begin{prooftree}
\AxiomC{$\mathsf{regT}(k) \in \mathrm{dom}(\mathit{ctx})$}
\AxiomC{$\mathit{val} = \mathrm{range}(\mathit{ctx}(\mathsf{regT}(k)))$}
\LeftLabel{$\textsc{Val}_2$}
\BinaryInfC{$\mathit{ctx},\, \mathsf{cursor}(\langle\rangle,\, k).\mathsf{values} \evalto \mathit{val}$}
\end{prooftree}

\begin{prooftree}
\AxiomC{$k_1 \in \mathrm{dom}(\mathit{ctx})$}
\AxiomC{$\mathit{ctx}(k_1),\, \mathsf{cursor}(\langle k_2, \dots, k_{n-1} \rangle,\, k_n).\mathsf{values}
    \evalto \mathit{val}$}
\LeftLabel{$\textsc{Val}_3$}
\BinaryInfC{$\mathit{ctx},\, \mathsf{cursor}(\langle k_1, k_2, \dots, k_{n-1} \rangle,\, k_n).\mathsf{values}
    \evalto \mathit{val}$}
\end{prooftree}
\caption{Rules for evaluating expressions.}\label{fig:expr-rules}
\end{figure*}

\section{Formal Semantics}\label{sec:semantics}

We now explain formally how to achieve the concurrent semantics outlined in Section~\ref{sec:composing}. The state of replica $p$ is described by $A_p$, a finite partial function. The evaluation rules of the command language inspect and modify this local state $A_p$, and they do not depend on $A_q$ (the state of any other replica $q \neq p$). The only communication between replicas occurs in the evaluation of the \textsf{yield} command, which we discuss later. For now, we concentrate on the execution of commands at a single replica $p$.

\subsection{Expression Evaluation}

Figure~\ref{fig:expr-rules} gives the rules for evaluating EXPR expressions in the command language, which are evaluated in the context of the local replica state $A_p$. The \textsc{Exec} rule formalizes the assumption that commands are executed sequentially. The \textsc{Let} rule allows the program to define a local variable, which is added to the local state (the notation $A_p[\,x \,\mapsto\, \mathit{cur}\,]$ denotes a partial function that is the same as $A_p$, except that $A_p(x) = \mathit{cur}$). The corresponding \textsc{Var} rule allows the program to retrieve the value of a previously defined variable.

The following rules in Figure~\ref{fig:expr-rules} show how an expression is evaluated to a cursor, which unambiguously identifies a particular position in a JSON document by describing a path from the root of the document tree to some branch or leaf node. A cursor consists only of immutable keys and identifiers, so it can be sent over the network to another replica, where it can be used to locate the same position in the document.

For example,
\[ \mathsf{cursor}(\langle \mathsf{mapT}(\mathsf{doc}), \mathsf{listT}(\text{``shopping''}) \rangle,\, \mathit{id}_1) \]
is a cursor representing the list element \verb|"eggs"| in Figure~\ref{fig:make-doc}, assuming that $\mathit{id}_1$ is the unique identifier of the operation that inserted this list element (we will discuss these identifiers in Section~\ref{sec:lamport-ts}). The cursor can be interpreted as a path through the local replica state structure $A_p$, read from left to right: starting from the \textsf{doc} map at the root, it traverses through the map entry ``shopping'' of type \textsf{listT}, and finishes with the list element that has identifier $\mathit{id}_1$.

In general, $\mathsf{cursor}(\langle k_1, \dots, k_{n-1} \rangle,\, k_n)$ consists of a (possibly empty) vector of keys $\langle k_1, \dots, k_{n-1} \rangle$, and a final key $k_n$ which is always present. $k_n$ can be thought of as the final element of the vector, with the distinction that it is not tagged with a datatype, whereas the elements of the vector are tagged with the datatype of the branch node being traversed, either \textsf{mapT} or \textsf{listT}.

The \textsc{Doc} rule in Figure~\ref{fig:expr-rules} defines the simplest cursor $\mathsf{cursor}(\langle\rangle,\, \mathsf{doc})$, referencing the root of the document using the special atom \textsf{doc}. The \textsc{Get} rule navigates a cursor to a particular key within a map. For example, the expression \verb|doc.get("shopping")| evaluates to $\mathsf{cursor}(\langle \mathsf{mapT}(\mathsf{doc}) \rangle,\, \text{``shopping''})$ by applying the \textsc{Doc} and \textsc{Get} rules. Note that the expression \verb|doc.get(...)| implicitly asserts that \textsf{doc} is of type \textsf{mapT}, and this assertion is encoded in the cursor.

The rules $\textsc{Idx}_{1 \dots 5}$ define how to evaluate the expression \verb|.idx(|$n$\verb|)|, moving the cursor to a particular element of a list. $\textsc{Idx}_1$ constructs a cursor representing the head of the list, and delegates to the subsequent rules to scan over the list. $\textsc{Idx}_2$ recursively descends the local state according to the vector of keys in the cursor. When the vector of keys is empty, the context $\mathit{ctx}$ is the subtree of $A_p$ that stores the list in question, and the rules $\textsc{Idx}_{3,4,5}$ iterate over that list until the desired element is found.

$\textsc{Idx}_5$ terminates the iteration when the index reaches zero, while $\textsc{Idx}_3$ moves to the next element and decrements the index, and $\textsc{Idx}_4$ skips over list elements that are marked as deleted. The structure resembles a linked list: each list element has a unique identifier $k$, and the partial function representing local state maps $\mathsf{next}(k)$ to the ID of the list element that follows $k$.

Deleted elements are never removed from the linked list structure, but only marked as deleted (they become so-called \emph{tombstones}). To this end, the local state maintains a \emph{presence set} $\mathsf{pres}(k)$ for the list element with ID $k$, which is the set of all operations that have asserted the existence of this list element. When a list element is deleted, the presence set is set to the empty set, which marks it as deleted; however, a concurrent operation that references the list element can cause the presence set to be come non-empty again (leading to the situations in Figures~\ref{fig:map-remove} and~\ref{fig:todo-item}). Rule $\textsc{Idx}_4$ handles list elements with an empty presence set by moving to the next list element without decrementing the index (i.e., not counting them as list elements).

The $\textsc{Keys}_{1,2,3}$ rules allow the application to inspect the set of keys in a map. This set is determined by examining the local state, and excluding any keys for which the presence set is empty (indicating that the key has been deleted).

Finally, the $\textsc{Val}_{1,2,3}$ rules allow the application to read the contents of a register at a particular cursor position, using a similar recursive rule structure as the \textsc{Idx} rules. A register is expressed using the \textsf{regT} type annotation in the local state. Although a replica can only assign a single value to a register, a register can nevertheless contain multiple values if multiple replicas concurrently assign values to it.

\subsection{Generating Operations}

When commands mutate the state of the document, they generate \emph{operations} that describe the mutation. In our semantics, a command never directly modifies the local replica state $A_p$, but only generates an operation. That operation is then immediately applied to $A_p$ so that it takes effect locally, and the same operation is also asynchronously broadcast to the other replicas.

\subsubsection{Lamport Timestamps}\label{sec:lamport-ts}

Every operation in our model is given a unique identifier, which is used in the local state and in cursors. Whenever an element is inserted into a list, or a value is assigned to a register, the new list element or register value is identified by the identifier of the operation that created it.

In order to generate globally unique operation identifiers without requiring synchronous coordination between replicas we use Lamport timestamps~\cite{Lamport:1978jq}. A Lamport timestamp is a pair $(c, p)$ where $p \in \mathrm{ReplicaID}$ is the unique identifier of the replica on which the edit is made (for example, a hash of its public key), and $c \in \mathbb{N}$ is a counter that is stored at each replica and incremented for every operation. Since each replica generates a strictly monotonically increasing sequence of counter values $c$, the pair $(c, p)$ is unique.

If a replica receives an operation with a counter value $c$ that is greater than the locally stored counter value, the local counter is increased to the value of the incoming counter. This ensures that if operation $o_1$ causally happened before $o_2$ (that is, the replica that generated $o_2$ had received and processed $o_1$ before $o_2$ was generated), then $o_2$ must have a greater counter value than $o_1$. Only concurrent operations can have equal counter values.

We can thus define a total ordering $<$ for Lamport timestamps:
\[ (c_1, p_1) < (c_2, p_2) \;\text{ iff }\; (c_1 < c_2) \vee (c_1 = c_2 \wedge p_1 < p_2). \]
If one operation happened before another, this ordering is consistent with causality (the earlier operation has a lower timestamp). If two operations are concurrent, their order according to $<$ is arbitrary but deterministic. This ordering property is important for our definition of the semantics of ordered lists.

\subsubsection{Operation Structure}

An operation is a tuple of the form
\begin{alignat*}{3}
& \mathsf{op}( \\
&& \mathit{id} &: \mathbb{N} \times \mathrm{ReplicaID}, \\
&& \mathit{deps} &: \mathcal{P}(\mathbb{N} \times \mathrm{ReplicaID}), \\
&& \mathit{cur} &: \mathsf{cursor}(\langle k_1, \dots, k_{n-1} \rangle,\, k_n), \\
&& \mathit{mut} &: \mathsf{insert}(v) \mid \mathsf{delete} \mid \mathsf{assign}(v) && \quad v: \mathrm{VAL} \\
& )
\end{alignat*}
where $\mathit{id}$ is the Lamport timestamp that uniquely identifies the operation, $\mathit{cur}$ is the cursor describing the position in the document being modified, and $\mathit{mut}$ is the mutation that was requested at the specified position.

$\mathit{deps}$ is the set of \emph{causal dependencies} of the operation. It is defined as follows: if operation $o_2$ was generated by replica $p$, then a causal dependency of $o_2$ is any operation $o_1$ that had already been applied on $p$ at the time when $o_2$ was generated. In this paper, we define $\mathit{deps}$ as the set of Lamport timestamps of all causal dependencies. In a real implementation, this set would become impracticably large, so a compact representation of causal history would be used instead -- for example, version vectors~\cite{ParkerJr:1983jb}, state vectors~\cite{Ellis:1989ue}, or dotted version vectors~\cite{Preguica:2012fx}. However, to avoid ambiguity in our semantics we give the dependencies as a simple set of operation IDs.

The purpose of the causal dependencies $\mathit{deps}$ is to impose a partial ordering on operations: an operation can only be applied after all operations that ``happened before'' it have been applied. In particular, this means that the sequence of operations generated at one particular replica will be applied in the same order at every other replica. Operations that are concurrent (i.e., where there is no causal dependency in either direction) can be applied in any order.

\subsubsection{Semantics of Generating Operations}

\begin{figure*}
\centering
\begin{prooftree}
\AxiomC{$A_p,\, \mathit{expr} \evalto \mathit{cur}$}
\AxiomC{$\mathit{val}: \mathrm{VAL}$}
\AxiomC{$A_p,\, \mathsf{makeOp}(\mathit{cur}, \mathsf{assign}(\mathit{val})) \evalto A_p'$}
\LeftLabel{\textsc{Make-Assign}}
\TrinaryInfC{$A_p,\, \mathit{expr} \,\text{ := }\, \mathit{val} \evalto A_p'$}
\end{prooftree}

\begin{prooftree}
\AxiomC{$A_p,\, \mathit{expr} \evalto \mathit{cur}$}
\AxiomC{$\mathit{val}: \mathrm{VAL}$}
\AxiomC{$A_p,\, \mathsf{makeOp}(\mathit{cur}, \mathsf{insert}(\mathit{val})) \evalto A_p'$}
\LeftLabel{\textsc{Make-Insert}}
\TrinaryInfC{$A_p,\, \mathit{expr}.\mathsf{insertAfter}(\mathit{val}) \evalto A_p'$}
\end{prooftree}

\begin{prooftree}
\AxiomC{$A_p,\, \mathit{expr} \evalto \mathit{cur}$}
\AxiomC{$A_p,\, \mathsf{makeOp}(\mathit{cur}, \mathsf{delete}) \evalto A_p'$}
\LeftLabel{\textsc{Make-Delete}}
\BinaryInfC{$A_p,\, \mathit{expr}.\mathsf{delete} \evalto A_p'$}
\end{prooftree}

\begin{prooftree}
\AxiomC{$\mathit{ctr} = \mathrm{max}(\{0\} \,\cup\, \{ c_i \mid (c_i, p_i) \in A_p(\mathsf{ops}) \}$}
\AxiomC{$A_p,\, \mathsf{apply}(\mathsf{op}((\mathit{ctr} + 1, p), A_p(\mathsf{ops}),
    \mathit{cur}, \mathit{mut})) \evalto A_p'$}
\LeftLabel{\textsc{Make-Op}}
\BinaryInfC{$A_p,\, \mathsf{makeOp}(\mathit{cur}, \mathit{mut}) \evalto A_p'$}
\end{prooftree}

\begin{prooftree}
\AxiomC{$A_p,\, \mathit{op} \evalto A_p'$}
\LeftLabel{\textsc{Apply-Local}}
\UnaryInfC{$A_p,\, \mathsf{apply}(\mathit{op}) \evalto A_p'[\,
    \mathsf{queue} \,\mapsto\, A_p'(\mathsf{queue}) \,\cup\, \{\mathit{op}\},\;
    \mathsf{ops} \,\mapsto\, A_p'(\mathsf{ops}) \,\cup\, \{\mathit{op.id}\}\,]$}
\end{prooftree}

\begin{prooftree}
\AxiomC{$\mathit{op} \in A_p(\mathsf{recv})$}
\AxiomC{$\mathit{op.id} \notin A_p(\mathsf{ops})$}
\AxiomC{$\mathit{op.deps} \subseteq A_p(\mathsf{ops})$}
\AxiomC{$A_p,\, \mathit{op} \evalto A_p'$}
\LeftLabel{\textsc{Apply-Remote}}
\QuaternaryInfC{$A_p,\, \mathsf{yield} \evalto
    A_p'[\,\mathsf{ops} \,\mapsto\, A_p'(\mathsf{ops}) \,\cup\, \{\mathit{op.id}\}\,]$}
\end{prooftree}

\begin{prooftree}
\AxiomC{}
\LeftLabel{\textsc{Send}}
\UnaryInfC{$A_p,\, \mathsf{yield} \evalto
    A_p[\,\mathsf{send} \,\mapsto\, A_p(\mathsf{send}) \,\cup\, A_p(\mathsf{queue})\,]$}
\end{prooftree}

\begin{prooftree}
\AxiomC{$q: \mathrm{ReplicaID}$}
\LeftLabel{\textsc{Recv}}
\UnaryInfC{$A_p,\, \mathsf{yield} \evalto
    A_p[\,\mathsf{recv} \,\mapsto\, A_p(\mathsf{recv}) \,\cup\, A_q(\mathsf{send})\,]$}
\end{prooftree}

\begin{prooftree}
\AxiomC{$A_p,\, \mathsf{yield} \evalto A_p'$}
\AxiomC{$A_p',\, \mathsf{yield} \evalto A_p''$}
\LeftLabel{\textsc{Yield}}
\BinaryInfC{$A_p,\, \mathsf{yield} \evalto A_p''$}
\end{prooftree}
\caption{Rules for generating, sending, and receiving operations.}
\label{fig:send-recv}
\end{figure*}

The evaluation rules for commands are given in Figure~\ref{fig:send-recv}. The \textsc{Make-Assign}, \textsc{Make-Insert} and \textsc{Make-Delete} rules define how these respective commands mutate the document: all three delegate to the \textsc{Make-Op} rule to generate and apply the operation. \textsc{Make-Op} generates a new Lamport timestamp by choosing a counter value that is 1 greater than any existing counter in $A_p(\mathsf{ops})$, the set of all operation IDs that have been applied to replica $p$.

\textsc{Make-Op} constructs an \textsf{op()} tuple of the form described above, and delegates to the \textsc{Apply-Local} rule to process the operation. \textsc{Apply-Local} does three things: it evaluates the operation to produce a modified local state $A_p'$, it adds the operation to the queue of generated operations $A_p(\mathsf{queue})$, and it adds the ID to the set of processed operations $A_p(\mathsf{ops})$.

The \textsf{yield} command, inspired by Burckhardt et al.~\cite{Burckhardt:2012jy}, performs network communication: sending and receiving operations to and from other replicas, and applying operations from remote replicas. The rules \textsc{Apply-Remote}, \textsc{Send}, \textsc{Recv} and \textsc{Yield} define the semantics of \textsf{yield}. Since any of these rules can be used to evaluate \textsf{yield}, their effect is nondeterministic, which models the asynchronicity of the network: a message sent by one replica arrives at another replica at some arbitrarily later point in time, and there is no message ordering guarantee in the network.

The \textsc{Send} rule takes any operations that were placed in $A_p(\mathsf{queue})$ by \textsc{Apply-Local} and adds them to a send buffer $A_p(\mathsf{send})$. Correspondingly, the \textsc{Recv} rule takes operations in the send buffer of replica $q$ and adds them to the receive buffer $A_p(\mathsf{recv})$ of replica $p$. This is the only rule that involves more than one replica, and it models all network communication.

Once an operation appears in the receive buffer $A_p(\mathsf{recv})$, the rule \textsc{Apply-Remote} may apply. Under the preconditions that the operation has not already been processed and that its causal dependencies are satisfied, \textsc{Apply-Remote} applies the operation in the same way as \textsc{Apply-Local}, and adds its ID to the set of processed operations $A_p(\mathsf{ops})$.

The actual document modifications are performed by applying the operations, which we discuss next.

\subsection{Applying Operations}

\begin{sidewaysfigure*}
\begin{prooftree}
\AxiomC{$\mathit{ctx},\, k_1 \evalto \mathit{child}$}
\AxiomC{$\mathit{child},\, \mathsf{op}(\mathit{id}, \mathit{deps},
    \mathsf{cursor}(\langle k_2, \dots, k_{n-1} \rangle,\, k_n), \mathit{mut}) \evalto \mathit{child}'$}
\AxiomC{$\mathit{ctx},\, \mathsf{addId}(k_1, \mathit{id}, \mathit{mut}) \evalto \mathit{ctx}'$}
\LeftLabel{\textsc{Descend}}
\TrinaryInfC{$\mathit{ctx},\, \mathsf{op}(\mathit{id}, \mathit{deps},
    \mathsf{cursor}(\langle k_1, k_2, \dots, k_{n-1} \rangle,\, k_n), \mathit{mut}) \evalto
    \mathit{ctx}' [\, k_1 \,\mapsto\, \mathit{child}' \,]$}
\end{prooftree}\vspace{6pt}

\begin{center}
\AxiomC{$k \in \mathrm{dom}(\mathit{ctx})$}
\LeftLabel{\textsc{Child-Get}}
\UnaryInfC{$\mathit{ctx},\, k \evalto \mathit{ctx}(k)$}
\DisplayProof\hspace{3em}
\AxiomC{$\mathsf{mapT}(k) \notin \mathrm{dom}(\mathit{ctx})$}
\LeftLabel{\textsc{Child-Map}}
\UnaryInfC{$\mathit{ctx},\, \mathsf{mapT}(k) \evalto \{\}$}
\DisplayProof\hspace{3em}
\AxiomC{$\mathsf{listT}(k) \notin \mathrm{dom}(\mathit{ctx})$}
\LeftLabel{\textsc{Child-List}}
\UnaryInfC{$\mathit{ctx},\, \mathsf{listT}(k) \evalto
    \{\,\mathsf{next}(\mathsf{head}) \,\mapsto\, \mathsf{tail} \,\}$}
\DisplayProof\proofSkipAmount
\end{center}\vspace{6pt}

\begin{center}
\AxiomC{$\mathsf{regT}(k) \notin \mathrm{dom}(\mathit{ctx})$}
\LeftLabel{\textsc{Child-Reg}}
\UnaryInfC{$\mathit{ctx},\, \mathsf{regT}(k) \evalto \{\}$}
\DisplayProof\hspace{3em}
\AxiomC{$\mathsf{pres}(k) \in \mathrm{dom}(\mathit{ctx})$}
\LeftLabel{$\textsc{Presence}_1$}
\UnaryInfC{$\mathit{ctx},\, \mathsf{pres}(k) \evalto \mathit{ctx}(\mathsf{pres}(k))$}
\DisplayProof\hspace{3em}
\AxiomC{$\mathsf{pres}(k) \notin \mathrm{dom}(\mathit{ctx})$}
\LeftLabel{$\textsc{Presence}_2$}
\UnaryInfC{$\mathit{ctx},\, \mathsf{pres}(k) \evalto \{\}$}
\DisplayProof\proofSkipAmount
\end{center}\vspace{6pt}

\begin{center}
\AxiomC{$\mathit{mut} \not= \mathsf{delete}$}
\AxiomC{$k_\mathit{tag} \in \{\mathsf{mapT}(k), \mathsf{listT}(k), \mathsf{regT}(k)\}$}
\AxiomC{$\mathit{ctx},\, \mathsf{pres}(k) \evalto \mathit{pres}$}
\LeftLabel{$\textsc{Add-ID}_1$}
\TrinaryInfC{$\mathit{ctx},\, \mathsf{addId}(k_\mathit{tag}, \mathit{id}, \mathit{mut}) \evalto
    \mathit{ctx}[\, \mathsf{pres}(k) \,\mapsto\, \mathit{pres} \,\cup\, \{\mathit{id}\} \,]$}
\DisplayProof\hspace{3em}
\AxiomC{$\mathit{mut} = \mathsf{delete}$}
\LeftLabel{$\textsc{Add-ID}_2$}
\UnaryInfC{$\mathit{ctx},\, \mathsf{addId}(k_\mathit{tag}, \mathit{id}, \mathit{mut}) \evalto \mathit{ctx}$}
\DisplayProof\proofSkipAmount
\end{center}\vspace{6pt}

\begin{prooftree}
\AxiomC{$\mathit{val} \not= \texttt{[]} \,\wedge\, \mathit{val} \not= \texttt{\string{\string}}$}
\AxiomC{$\mathit{ctx},\, \mathsf{clear}(\mathit{deps}, \mathsf{regT}(k)) \evalto \mathit{ctx}',\, \mathit{pres}$}
\AxiomC{$\mathit{ctx}',\, \mathsf{addId}(\mathsf{regT}(k), \mathit{id}, \mathsf{assign}(\mathit{val}))
    \evalto \mathit{ctx}''$}
\AxiomC{$\mathit{ctx}'',\, \mathsf{regT}(k) \evalto \mathit{child}$}
\LeftLabel{\textsc{Assign}}
\QuaternaryInfC{$\mathit{ctx},\, \mathsf{op}(\mathit{id}, \mathit{deps},
    \mathsf{cursor}(\langle\rangle,\, k), \mathsf{assign}(\mathit{val})) \evalto
    \mathit{ctx}''[\, \mathsf{regT}(k) \,\mapsto\,
    \mathit{child}[\, \mathit{id} \,\mapsto\, \mathit{val} \,]\,]$}
\end{prooftree}\vspace{6pt}

\begin{prooftree}
\AxiomC{$\mathit{val} = \texttt{\string{\string}}$}
\AxiomC{$\mathit{ctx},\, \mathsf{clearElem}(\mathit{deps}, k) \evalto \mathit{ctx}',\, \mathit{pres}$}
\AxiomC{$\mathit{ctx}',\, \mathsf{addId}(\mathsf{mapT}(k), \mathit{id}, \mathsf{assign}(\mathit{val}))
    \evalto \mathit{ctx}''$}
\AxiomC{$\mathit{ctx}'',\, \mathsf{mapT}(k) \evalto \mathit{child}$}
\LeftLabel{\textsc{Empty-Map}}
\QuaternaryInfC{$\mathit{ctx},\, \mathsf{op}(\mathit{id}, \mathit{deps},
    \mathsf{cursor}(\langle\rangle,\, k), \mathsf{assign}(\mathit{val})) \evalto
    \mathit{ctx}''[\, \mathsf{mapT}(k) \,\mapsto\, \mathit{child} \,]$}
\end{prooftree}\vspace{6pt}

\begin{prooftree}
\AxiomC{$\mathit{val} = \texttt{[]}$}
\AxiomC{$\mathit{ctx},\, \mathsf{clearElem}(\mathit{deps}, k) \evalto \mathit{ctx}',\, \mathit{pres}$}
\AxiomC{$\mathit{ctx}',\, \mathsf{addId}(\mathsf{listT}(k), \mathit{id}, \mathsf{assign}(\mathit{val}))
    \evalto \mathit{ctx}''$}
\AxiomC{$\mathit{ctx}'',\, \mathsf{listT}(k) \evalto \mathit{child}$}
\LeftLabel{\textsc{Empty-List}}
\QuaternaryInfC{$\mathit{ctx},\, \mathsf{op}(\mathit{id}, \mathit{deps},
    \mathsf{cursor}(\langle\rangle,\, k), \mathsf{assign}(\mathit{val})) \evalto
    \mathit{ctx}''[\, \mathsf{listT}(k) \,\mapsto\, \mathit{child} \,]$}
\end{prooftree}\vspace{6pt}

\begin{prooftree}
\AxiomC{$\mathit{ctx}(\mathsf{next}(\mathit{prev})) = \mathit{next}$}
\AxiomC{$\mathit{next} < \mathit{id} \,\vee\, \mathit{next} = \mathsf{tail}$}
\AxiomC{$\mathit{ctx},\, \mathsf{op}(\mathit{id}, \mathit{deps},
    \mathsf{cursor}(\langle\rangle,\, \mathit{id}), \mathsf{assign}(\mathit{val})) \evalto \mathit{ctx}'$}
\LeftLabel{$\textsc{Insert}_1$}
\TrinaryInfC{$\mathit{ctx},\, \mathsf{op}(\mathit{id}, \mathit{deps},
    \mathsf{cursor}(\langle\rangle,\, \mathit{prev}), \mathsf{insert}(\mathit{val})) \evalto
    \mathit{ctx}'[\,\mathsf{next}(\mathit{prev}) \,\mapsto\, \mathit{id},\;
    \mathsf{next}(\mathit{id}) \,\mapsto\, \mathit{next}\,]$}
\end{prooftree}\vspace{6pt}

\begin{prooftree}
\AxiomC{$\mathit{ctx}(\mathsf{next}(\mathit{prev})) = \mathit{next}$}
\AxiomC{$\mathit{id} < \mathit{next}$}
\AxiomC{$ctx,\, \mathsf{op}(\mathit{id}, \mathit{deps},
    \mathsf{cursor}(\langle\rangle,\, \mathit{next}), \mathsf{insert}(\mathit{val})) \evalto \mathit{ctx}'$}
\LeftLabel{$\textsc{Insert}_2$}
\TrinaryInfC{$\mathit{ctx},\, \mathsf{op}(\mathit{id}, \mathit{deps},
    \mathsf{cursor}(\langle\rangle,\, \mathit{prev}), \mathsf{insert}(\mathit{val})) \evalto \mathit{ctx}'$}
\end{prooftree}
\caption{Rules for applying insertion and assignment operations to update the state of a replica.}\label{fig:operation-rules}
\end{sidewaysfigure*}

Figure~\ref{fig:operation-rules} gives the rules that apply an operation $\mathsf{op}$ to a context $\mathit{ctx}$, producing an updated context $\mathit{ctx}'$. The context is initially the replica state $A_p$, but may refer to subtrees of the state as rules are recursively applied. These rules are used by \textsc{Apply-Local} and \textsc{Apply-Remote} to perform the state updates on a document.

When the operation cursor's vector of keys is non-empty, the \textsc{Descend} rule first applies. It recursively descends the document tree by following the path described by the keys. If the tree node already exists in the local replica state, \textsc{Child-Get} finds it, otherwise \textsc{Child-Map} and \textsc{Child-List} create an empty map or list respectively.

The \textsc{Descend} rule also invokes $\textsc{Add-ID}_{1,2}$ at each tree node along the path described by the cursor, adding the operation ID to the presence set $\mathsf{pres}(k)$ to indicate that the subtree includes a mutation made by this operation.

The remaining rules in Figure~\ref{fig:operation-rules} apply when the vector of keys in the cursor is empty, which is the case when descended to the context of the tree node to which the mutation applies. The \textsc{Assign} rule handles assignment of a primitive value to a register, \textsc{Empty-Map} handles assignment where the value is the empty map literal \verb|{}|, and \textsc{Empty-List} handles assignment of the empty list \verb|[]|. These three rules for \textsf{assign} have a similar structure: first clearing the prior value at the cursor (as discussed in the next section), then adding the operation ID to the presence set, and finally incorporating the new value into the tree of local state.

The $\textsc{Insert}_{1,2}$ rules handle insertion of a new element into an ordered list. In this case, the cursor refers to the list element $\mathit{prev}$, and the new element is inserted after that position in the list. $\textsc{Insert}_1$ performs the insertion by manipulating the linked list structure. $\textsc{Insert}_2$ handles the case of multiple replicas concurrently inserting list elements at the same position, and uses the ordering relation $<$ on Lamport timestamps to consistently determine the insertion point. Our approach for handling insertions is based on the RGA algorithm~\cite{Roh:2011dw}. We show later that these rules ensure all replicas converge towards the same state.

\subsubsection{Clearing Prior State}

\begin{figure*}
\begin{prooftree}
\AxiomC{$\mathit{ctx},\, \mathsf{clearElem}(\mathit{deps}, k) \evalto \mathit{ctx}',\, \mathit{pres}$}
\LeftLabel{\textsc{Delete}}
\UnaryInfC{$\mathit{ctx},\, \mathsf{op}(\mathit{id}, \mathit{deps},
    \mathsf{cursor}(\langle\rangle,\, k), \mathsf{delete}) \evalto \mathit{ctx}'$}
\end{prooftree}

\begin{prooftree}
\AxiomC{$\mathit{ctx},\, \mathsf{clearAny}(\mathit{deps}, k) \evalto \mathit{ctx}', \mathit{pres}_1$}
\AxiomC{$\mathit{ctx}',\, \mathsf{pres}(k) \evalto \mathit{pres}_2$}
\AxiomC{$\mathit{pres}_3 = \mathit{pres}_1 \,\cup\, \mathit{pres}_2 \setminus \mathit{deps}$}
\LeftLabel{\textsc{Clear-Elem}}
\TrinaryInfC{$\mathit{ctx},\, \mathsf{clearElem}(\mathit{deps}, k) \evalto
    \mathit{ctx}' [\, \mathsf{pres}(k) \,\mapsto\, \mathit{pres}_3 \,],\, \mathit{pres}_3$}
\end{prooftree}

\begin{prooftree}
\AxiomC{$\begin{matrix}
    \mathit{ctx},\, \mathsf{clear}(\mathit{deps}, \mathsf{mapT}(k)) \\
    \evalto \mathit{ctx}_1,\, \mathit{pres}_1 \end{matrix}$}
\AxiomC{$\begin{matrix}
    \mathit{ctx}_1,\, \mathsf{clear}(\mathit{deps}, \mathsf{listT}(k)) \\
    \evalto \mathit{ctx}_2,\, \mathit{pres}_2 \end{matrix}$}
\AxiomC{$\begin{matrix}
    \mathit{ctx}_2,\, \mathsf{clear}(\mathit{deps}, \mathsf{regT}(k)) \\
    \evalto \mathit{ctx}_3,\, \mathit{pres}_3 \end{matrix}$}
\LeftLabel{\textsc{Clear-Any}}
\TrinaryInfC{$\mathit{ctx},\, \mathsf{clearAny}(\mathit{deps}, k) \evalto
    \mathit{ctx}_3,\, \mathit{pres}_1 \,\cup\, \mathit{pres}_2 \,\cup\, \mathit{pres}_3$}
\end{prooftree}

\begin{prooftree}
\AxiomC{$k \notin \mathrm{dom}(\mathit{ctx})$}
\LeftLabel{\textsc{Clear-None}}
\UnaryInfC{$\mathit{ctx},\, \mathsf{clear}(\mathit{deps}, k) \evalto \mathit{ctx},\, \{\}$}
\end{prooftree}

\begin{prooftree}
\AxiomC{$\mathsf{regT}(k) \in \mathrm{dom}(\mathit{ctx})$}
\AxiomC{$\mathit{concurrent} = \{ \mathit{id} \mapsto v \mid
    (\mathit{id} \mapsto v) \in \mathit{ctx}(\mathsf{regT}(k))
    \,\wedge\, \mathit{id} \notin \mathit{deps} \}$}
\LeftLabel{\textsc{Clear-Reg}}
\BinaryInfC{$\mathit{ctx},\, \mathsf{clear}(\mathit{deps}, \mathsf{regT}(k)) \evalto
    \mathit{ctx}[\, \mathsf{regT}(k) \,\mapsto\, \mathit{concurrent} \,],\, \mathrm{dom}(\mathit{concurrent})$}
\end{prooftree}

\begin{prooftree}
\AxiomC{$\mathsf{mapT}(k) \in \mathrm{dom}(\mathit{ctx})$}
\AxiomC{$\mathit{ctx}(\mathsf{mapT}(k)),\, \mathsf{clearMap}(\mathit{deps}, \{\}) \evalto
    \mathit{cleared},\, \mathit{pres}$}
\LeftLabel{$\textsc{Clear-Map}_1$}
\BinaryInfC{$\mathit{ctx},\, \mathsf{clear}(\mathit{deps}, \mathsf{mapT}(k)) \evalto
    \mathit{ctx} [\, \mathsf{mapT}(k) \,\mapsto\, \mathit{cleared} \,],\, \mathit{pres}$}
\end{prooftree}

\begin{prooftree}
\AxiomC{$\begin{matrix}
    k \in \mathrm{keys}(\mathit{ctx}) \\
    \wedge\, k \notin \mathit{done} \end{matrix}$}
\AxiomC{$\begin{matrix}
    \mathit{ctx},\, \mathsf{clearElem}(\mathit{deps}, k) \\
    \evalto \mathit{ctx}', \mathit{pres}_1 \end{matrix}$}
\AxiomC{$\begin{matrix}
    \mathit{ctx}',\, \mathsf{clearMap}(\mathit{deps}, \mathit{done} \cup \{k\}) \\
    \evalto \mathit{ctx}'',\, \mathit{pres}_2 \end{matrix}$}
\LeftLabel{$\textsc{Clear-Map}_2$}
\TrinaryInfC{$\mathit{ctx},\, \mathsf{clearMap}(\mathit{deps}, \mathit{done})
    \evalto \mathit{ctx}'',\, \mathit{pres}_1 \,\cup\, \mathit{pres}_2$}
\end{prooftree}

\begin{prooftree}
\AxiomC{$\mathit{done} = \mathrm{keys}(\mathit{ctx})$}
\LeftLabel{$\textsc{Clear-Map}_3$}
\UnaryInfC{$\mathit{ctx},\, \mathsf{clearMap}(\mathit{deps}, \mathit{done}) \evalto \mathit{ctx},\, \{\}$}
\end{prooftree}

\begin{prooftree}
\AxiomC{$\mathsf{listT}(k) \in \mathrm{dom}(\mathit{ctx})$}
\AxiomC{$\mathit{ctx}(\mathsf{listT}(k)),\, \mathsf{clearList}(\mathit{deps}, \mathsf{head})
    \evalto \mathit{cleared},\, \mathit{pres}$}
\LeftLabel{$\textsc{Clear-List}_1$}
\BinaryInfC{$\mathit{ctx},\, \mathsf{clear}(\mathit{deps}, \mathsf{listT}(k)) \evalto
    \mathit{ctx}[\, \mathsf{listT}(k) \,\mapsto\, \mathit{cleared} \,],\, \mathit{pres}$}
\end{prooftree}

\begin{prooftree}
\AxiomC{$\begin{matrix}
    k \not= \mathsf{tail} \,\wedge\\
    \mathit{ctx}(\mathsf{next}(k)) = \mathit{next} \end{matrix}$}
\AxiomC{$\begin{matrix}
    \mathit{ctx},\, \mathsf{clearElem}(\mathit{deps}, k) \\
    \evalto \mathit{ctx}', \mathit{pres}_1 \end{matrix}$}
\AxiomC{$\begin{matrix}
    \mathit{ctx}',\, \mathsf{clearList}(\mathit{deps}, \mathit{next}) \\
    \evalto \mathit{ctx}'', \mathit{pres}_2 \end{matrix}$}
\LeftLabel{$\textsc{Clear-List}_2$}
\TrinaryInfC{$\mathit{ctx},\, \mathsf{clearList}(\mathit{deps}, k) \evalto
    \mathit{ctx}'',\, \mathit{pres}_1 \,\cup\, \mathit{pres}_2$}
\end{prooftree}

\begin{prooftree}
\AxiomC{$k = \mathsf{tail}$}
\LeftLabel{$\textsc{Clear-List}_3$}
\UnaryInfC{$\mathit{ctx},\, \mathsf{clearList}(\mathit{deps}, k) \evalto \mathit{ctx},\, \{\}$}
\end{prooftree}
\caption{Rules for applying deletion operations to update the state of a replica.}\label{fig:clear-rules}
\end{figure*}

Assignment and deletion operations require that prior state (the value being overwritten or deleted) is cleared, while also ensuring that concurrent modifications are not lost, as illustrated in Figure~\ref{fig:map-remove}. The rules to handle this clearing process are given in Figure~\ref{fig:clear-rules}. Intuitively, the effect of clearing something is to reset it to its empty state by undoing any operations that causally precede the current operation, while leaving the effect of any concurrent operations untouched.

A \textsf{delete} operation can be used to delete either an element from an ordered list or a key from a map, depending on what the cursor refers to. The \textsc{Delete} rule shows how this operation is evaluated by delegating to \textsc{Clear-Elem}. In turn, \textsc{Clear-Elem} uses \textsc{Clear-Any} to clear out any data with a given key, regardless of whether it is of type \textsf{mapT}, \textsf{listT} or \textsf{regT}, and also updates the presence set to include any nested operation IDs, but exclude any operations in $\mathit{deps}$.

The premises of \textsc{Clear-Any} are satisfied by $\textsc{Clear-Map}_1$, $\textsc{Clear-List}_1$ and \textsc{Clear-Reg} if the respective key appears in $\mathit{ctx}$, or by \textsc{Clear-None} (which does nothing) if the key is absent.

As defined by the \textsc{Assign} rule, a register maintains a mapping from operation IDs to values. \textsc{Clear-Reg} updates a register by removing all operation IDs that appear in $\mathit{deps}$ (i.e., which causally precede the clearing operation), but retaining all operation IDs that do not appear in $\mathit{deps}$ (from assignment operations that are concurrent with the clearing operation).

Clearing maps and lists takes a similar approach: each element of the map or list is recursively cleared using \textsf{clearElem}, and presence sets are updated to exclude $\mathit{deps}$. Thus, any list elements or map entries whose modifications causally precede the clearing operation will end up with empty presence sets, and thus be considered deleted. Any map or list elements containing operations that are concurrent with the clearing operation are preserved.

\subsection{Convergence}\label{sec:convergence}

As outlined in Section~\ref{sec:intro-replication}, we require that all replicas automatically converge towards the same state -- a key requirement of a CRDT. We now formalize this notion, and show that the rules in Figures~\ref{fig:expr-rules} to~\ref{fig:clear-rules} satisfy this requirement.

\begin{definition}[valid execution]\label{def:valid-exec}
A \emph{valid execution} is a set of operations generated by a set of replicas $\{p_1, \dots, p_k\}$, each reducing a sequence of commands $\langle \mathit{cmd}_1 \mathbin{;} \dots \mathbin{;} \mathit{cmd}_n \rangle$ without getting stuck.
\end{definition}

A reduction gets stuck if there is no application of rules in which all premises are satisfied. For example, the $\textsc{Idx}_{3,4}$ rules in Figure~\ref{fig:expr-rules} get stuck if $\mathsf{idx}(n)$ tries to iterate past the end of a list, which would happen if $n$ is greater than the number of non-deleted elements in the list; in a real implementation this would be a runtime error. By constraining valid executions to those that do not get stuck, we ensure that operations only refer to list elements that actually exist.

Note that it is valid for an execution to never perform any network communication, either because it never invokes the \textsf{yield} command, or because the nondeterministic execution of \textsf{yield} never applies the \textsc{Recv} rule. We need only a replica's local state to determine whether reduction gets stuck.

\begin{definition}[history]\label{def:history}
A \emph{history} is a sequence of operations in the order it was applied at one particular replica $p$ by application of the rules \textsc{Apply-Local} and \textsc{Apply-Remote}.
\end{definition}

Since the evaluation rules sequentially apply one operation at a time at a given replica, the order is well-defined. Even if two replicas $p$ and $q$ applied the same set of operations, i.e. if $A_p(\mathsf{ops}) = A_q(\mathsf{ops})$, they may have applied any concurrent operations in a different order. Due to the premise $\mathit{op.deps} \subseteq A_p(\mathsf{ops})$ in \textsc{Apply-Remote}, histories are consistent with causality: if an operation has causal dependencies, it appears at some point after those dependencies in the history.

\begin{definition}[document state]\label{def:doc-state}
The \emph{document state} of a replica $p$ is the subtree of $A_p$ containing the document: that is, $A_p(\mathsf{mapT}(\mathsf{doc}))$ or $A_p(\mathsf{listT}(\mathsf{doc}))$ or $A_p(\mathsf{regT}(\mathsf{doc}))$, whichever is defined.
\end{definition}

$A_p$ contains variables defined with \textsf{let}, which are local to one replica, and not part of the replicated state. The definition of document state excludes these variables.

\begin{convergence-thm}
For any two replicas $p$ and $q$ that participated in a valid execution, if $A_p(\mathsf{ops}) = A_q(\mathsf{ops})$, then $p$ and $q$ have the same document state.
\end{convergence-thm}

This theorem is proved in the appendix. It formalizes the safety property of convergence: if two replicas have processed the same set of operations, possibly in a different order, then they are in the same state. In combination with a liveness property, namely that every replica eventually processes all operations, we obtain the desired notion of convergence: all replicas eventually end up in the same state.

The liveness property depends on assumptions of replicas invoking \textsf{yield} sufficiently often, and all nondeterministic rules for \textsf{yield} being chosen fairly. We will not formalize the liveness property in this paper, but assert that it can usually be provided in practice, as network interruptions are usually of finite duration.

\section{Conclusions and Further Work}

In this paper we demonstrated how to compose CRDTs for ordered lists, maps and registers into a compound CRDT with a JSON data model. It supports arbitrarily nested lists and maps, and it allows replicas to make arbitrary changes to the data without waiting for network communication. Replicas asynchronously send mutations to other replicas in the form of operations. Concurrent operations are commutative, which ensures that replicas converge towards the same state without requiring application-specific conflict resolution logic.

This work focused on the formal semantics of the JSON CRDT, represented as a mathematical model. We are also working on a practical implementation of the algorithm, and will report on its performance characteristics in follow-on work.

Our principle of not losing input due to concurrent modifications appears at first glance to be reasonable, but as illustrated in Figure~\ref{fig:todo-item}, it leads to merged document states that may be surprising to application programmers who are more familiar with sequential programs. Further work will be needed to understand the expectations of application programmers, and to design data structures that are minimally surprising under concurrent modification. It may turn out that a schema language will be required to support more complex applications. A schema language could also support semantic annotations, such as indicating that a number should be treated as a counter rather than a register.

The CRDT defined in this paper supports insertion, deletion and assignment operations. In addition to these, it would be useful to support a \emph{move} operation (to change the order of elements in an ordered list, or to move a subtree from one position in a document to another) and an \emph{undo} operation. Moreover, garbage collection (tombstone removal) is required in order to prevent unbounded growth of the data structure. We plan to address these missing features in future work.

\section*{Acknowledgements}

This research was supported by a grant from The Boeing Company. Thank you to Dominic Orchard, Diana Vasile, and the anonymous reviewers for comments that improved this paper.

\bibliographystyle{IEEEtran}

\vfill

\begin{IEEEbiography}[{\includegraphics[width=1in]{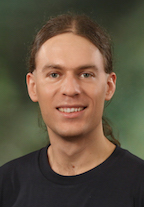}}]{Martin Kleppmann}
is a Research Associate in the Computer Laboratory at the University of Cambridge. His current research project, TRVE Data, is working towards better security and privacy in cloud applications by applying end-to-end encryption to collaboratively editable application data. His book \emph{Designing Data-Intensive Applications} was published by O'Reilly Media in 2017. Previously, he worked as a software engineer and entrepreneur at several internet companies, including Rapportive and LinkedIn.
\end{IEEEbiography}

\begin{IEEEbiography}[{\includegraphics[width=1in]{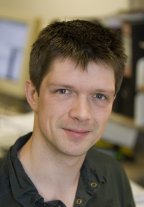}}]{Alastair R. Beresford}
is a Senior Lecturer in the Computer Laboratory at the University of Cambridge. His research work explores the security and privacy of large-scale distributed systems, with a particular focus on networked mobile devices such as smartphones, tablets and laptops. He looks at the security and privacy of the devices themselves, as well as the security and privacy problems induced by the interaction between mobile devices and cloud-based Internet services.
\end{IEEEbiography}

\ifincludeappendix
\clearpage
\appendix[Proof of Convergence]\label{sec:proof}

\begin{theorem}\label{thm:convergence}
For any two replicas $p$ and $q$ that participated in a valid execution, if $A_p(\mathsf{ops}) = A_q(\mathsf{ops})$, then $p$ and $q$ have the same document state.
\end{theorem}

\begin{proof}
Consider the histories $H_p$ and $H_q$ at $p$ and $q$ respectively (see Definition~\ref{def:history}). The rules \textsc{Apply-Local} and \textsc{Apply-Remote} maintain the invariant that an operation is added to $A_p(\mathsf{ops})$ or $A_q(\mathsf{ops})$ if and only if it was applied to the document state at $p$ or $q$. Thus, $A_p(\mathsf{ops}) = A_q(\mathsf{ops})$ iff $H_p$ and $H_q$ contain the same set of operations (potentially ordered differently).

The history $H_p$ at replica $p$ is a sequence of $n$ operations: $H_p = \langle o_1, \dots, o_n \rangle$, and the document state at $p$ is derived from $H_p$ by starting in the empty state and applying the operations in order. Likewise, the document state at $q$ is derived from $H_q$, which is a permutation of $H_p$. Both histories must be consistent with causality, i.e. for all $i$ with $1 \le i \le n$, we require $o_i.\mathit{deps} \subseteq \{o_j.\mathit{id} \mid 1 \le j < i\}$. The causality invariant is maintained by the \textsc{Apply-*} rules.

We can prove the theorem by induction over the length of history $n$.

\emph{Base case:} An empty history with $n=0$ describes the empty document state. The empty document is always the same, and so any two replicas that have not executed any operations are by definition in the same state.

\emph{Induction step:} Given histories $H_p$ and $H_q$ of length $n$, such that $H_p = \langle o_1, \dots, o_n \rangle$ and $H_q$ is a permutation of $H_p$, and such that applying $H_p$ results in the same document state as applying $H_q$, we can construct new histories $H_p'$ and $H_q'$ of length $n+1$ by inserting a new operation $o_{n+1}$ at any causally ready position in $H_p$ or $H_q$ respectively. We must then show that for all the histories $H_p'$ and $H_q'$ constructed this way, applying the sequence of operations in order results in the same document state.

In order to prove the induction step, we examine the insertion of $o_{n+1}$ into $H_p$ and $H_q$. Each history can be split into a prefix, which is the minimal subsequence $\langle o_1, \dots, o_j \rangle$ such that $o_{n+1}.\mathit{deps} \subseteq \{o_1.\mathit{id}, \dots, o_j.\mathit{id}\}$, and a suffix, which is the remaining subsequence $\langle o_{j+1}, \dots, o_n \rangle$. The prefix contains all operations that causally precede $o_{n+1}$, and possibly some operations that are concurrent with $o_{n+1}$; the suffix contains only operations that are concurrent with $o_{n+1}$. The earliest position where $o_{n+1}$ can be inserted into the history is between the prefix and the suffix; the latest position is at the end of the suffix; or it could be inserted at any point within the suffix.

We need to show that the effect on the document state is the same, regardless of the position at which $o_{n+1}$ is inserted, and regardless of whether it is inserted into $H_p$ or $H_q$. We do this in Lemma~\ref{lem:op-commute} by showing that $o_{n+1}$ is commutative with respect to all operations in the suffix, i.e. with respect to any operations that are concurrent to $o_{n+1}$.
\end{proof}

Before we can prove the commutativity of operations, we must first define some more terms and prove some preliminary lemmas.

\begin{definition}[appearing after]
In the ordered list $\mathit{ctx}$, list element $k_j$ \emph{appears after} list element $k_1$ if there exists a (possibly empty) sequence of list elements $k_2, \dots, k_{j-1}$ such that for all $i$ with $1 \le i < j$, $\mathit{ctx}(\mathsf{next}(k_i)) = k_{i+1}$. Moreover, we say $k_j$ appears \emph{immediately after} $k_1$ if that sequence is empty, i.e. if $\mathit{ctx}(\mathsf{next}(k_1)) = k_j$.
\end{definition}

The definition of \emph{appearing after} corresponds to the order in which the \textsc{Idx} rules iterate over the list.

\begin{lemma}\label{lem:list-after}
If $k_2$ appears after $k_1$ in an ordered list, and the list is mutated according to the evaluation rules, $k_2$ also appears after $k_1$ in all later document states.
\end{lemma}

\begin{proof}
The only rule that modifies the \textsf{next} pointers in the context is $\textsc{Insert}_1$, and it inserts a new list element between two existing list elements (possibly \textsf{head} and/or \textsf{tail}). This modification preserves the appears-after relationship between any two existing list elements. Since no other rule affects the list order, appears-after is always preserved.
\end{proof}

Note that deletion of an element from a list does not remove it from the sequence of \textsf{next} pointers, but only clears its presence set $\mathsf{pres}(k)$.

\begin{lemma}\label{lem:insert-between}
If one replica inserts a list element $k_\mathit{new}$ between $k_1$ and $k_2$, i.e. if $k_\mathit{new}$ appears after $k_1$ in the list and $k_2$ appears after $k_\mathit{new}$ in the list on the source replica after applying \textsc{Apply-Local}, then $k_\mathit{new}$ appears after $k_1$ and $k_2$ appears after $k_\mathit{new}$ on every other replica where that operation is applied.
\end{lemma}

\begin{proof}
The rules for generating list operations ensure that $k_1$ is either \textsf{head} or an operation identifier, and $k_2$ is either \textsf{tail} or an operation identifier.

When the insertion operation is generated using the \textsc{Make-Op} rule, its operation identifier is given a counter value $\mathit{ctr}$ that is greater than the counter of any existing operation ID in $A_p(\mathsf{ops})$. If $k_2$ is an operation identifier, we must have $k_2 \in A_p(\mathsf{ops})$, since both \textsc{Apply-Local} and \textsc{Apply-Remote} add operation IDs to $A_p(\mathsf{ops})$ when applying an insertion. Thus, either $k_2 < k_\mathit{new}$ under the ordering relation $<$ for Lamport timestamps, or $k_2 = \mathsf{tail}$.

When the insertion operation is applied on another replica using \textsc{Apply-Remote} and $\textsc{Insert}_{1,2}$, $k_2$ appears after $k_1$ on that replica (by Lemma~\ref{lem:list-after} and causality). The cursor of the operation is $\mathsf{cursor}(\langle\dots\rangle, k_1)$, so the rules start iterating the list at $k_1$, and therefore $k_\mathit{new}$ is inserted at some position after $k_1$.

If other concurrent insertions occurred between $k_1$ and $k_2$, their operation ID may be greater than or less than $k_\mathit{new}$, and thus either $\textsc{Insert}_1$ or $\textsc{Insert}_2$ may apply. In particular, $\textsc{Insert}_2$ skips over any list elements whose Lamport timestamp is greater than $k_\mathit{new}$. However, we know that $k_2 < k_\mathit{new} \vee k_2 = \mathsf{tail}$, and so $\textsc{Insert}_1$ will apply with $\mathit{next}=k_2$ at the latest. The $\textsc{Insert}_{1,2}$ rules thus never iterate past $k_2$, and thus $k_\mathit{new}$ is never inserted at a list position that appears after $k_2$.
\end{proof}

\begin{definition}[common ancestor]\label{def:common-ancestor}
In a history $H$, the \emph{common ancestor} of two concurrent operations $o_r$ and $o_s$ is the latest document state that causally precedes both $o_r$ and $o_s$.
\end{definition}

The common ancestor of $o_r$ and $o_s$ can be defined more formally as the document state resulting from applying a sequence of operations $\langle o_1, \dots, o_j \rangle$ that is the shortest prefix of $H$ that satisfies $(o_r.\mathit{deps} \cap o_s.\mathit{deps}) \subseteq \{o_1.\mathit{id}, \dots, o_j.\mathit{id}\}$.

\begin{definition}[insertion interval]\label{def:insert-interval}
Given two concurrent operations $o_r$ and $o_s$ that insert into the same list, the \emph{insertion interval} of $o_r$ is the pair of keys $(k_r^\mathrm{before}, k_r^\mathrm{after})$ such that $o_r.\mathit{id}$ appears after $k_r^\mathrm{before}$ when $o_r$ has been applied, $k_r^\mathrm{after}$ appears after $o_r.\mathit{id}$ when $o_r$ has been applied, and $k_r^\mathrm{after}$ appears immediately after $k_r^\mathrm{before}$ in the common ancestor of $o_r$ and $o_s$. The insertion interval of $o_s$ is the pair of keys $(k_s^\mathrm{before}, k_s^\mathrm{after})$ defined similarly.
\end{definition}

It may be the case that $k_r^\mathrm{before}$ or $k_s^\mathrm{before}$ is \textsf{head}, and that $k_r^\mathrm{after}$ or $k_s^\mathrm{after}$ is \textsf{tail}.

\begin{lemma}\label{lem:insert-conflict}
For any two concurrent insertion operations $o_r, o_s$ in a history $H$, if $o_r.\mathit{cur} = o_s.\mathit{cur}$, then the order at which the inserted elements appear in the list after applying $H$ is deterministic and independent of the order of $o_r$ and $o_s$ in $H$.
\end{lemma}

\begin{proof}
Without loss of generality, assume that $o_r.\mathit{id} < o_s.\mathit{id}$ according to the ordering relation on Lamport timestamps. (If the operation ID of $o_r$ is greater than that of $o_s$, the two operations can be swapped in this proof.) We now distinguish the two possible orders of applying the operations:

\begin{enumerate}
\item $o_r$ is applied before $o_s$ in $H$. Thus, at the time when $o_s$ is applied, $o_r$ has already been applied. When applying $o_s$, since $o_r$ has a lesser operation ID, the rule $\textsc{Insert}_1$ applies with $\mathit{next} = o_r.\mathit{id}$ at the latest, so the insertion position of $o_s$ must appear before $o_r$. It is not possible for $\textsc{Insert}_2$ to skip past $o_r$.

\item $o_s$ is applied before $o_r$ in $H$. Thus, at the time when $o_r$ is applied, $o_s$ has already been applied. When applying $o_r$, the rule $\textsc{Insert}_2$ applies with $\mathit{next} = o_s.\mathit{id}$, so the rule skips past $o_s$ and inserts $o_r$ at a position after $o_s$. Moreover, any list elements that appear between $o_s.\mathit{cur}$ and $o_s$ at the time of inserting $o_r$ must have a Lamport timestamp greater than $o_s.\mathit{id}$, so $\textsc{Insert}_2$ also skips over those list elements when inserting $o_r$. Thus, the insertion position of $o_r$ must be after $o_s$.
\end{enumerate}

Thus, the insertion position of $o_r$ appears after the insertion position of $o_s$, regardless of the order in which the two operations are applied. The ordering depends only on the operation IDs, and since these IDs are fixed at the time the operations are generated, the list order is determined by the IDs.
\end{proof}

\begin{lemma}\label{lem:insert-reorder}
In an operation history $H$, an insertion operation is commutative with respect to concurrent insertion operations to the same list.
\end{lemma}

\begin{proof}
Given any two concurrent insertion operations $o_r, o_s$ in $H$, we must show that the document state does not depend on the order in which $o_r$ and $o_s$ are applied.

Either $o_r$ and $o_s$ have the same insertion interval as defined in Definition~\ref{def:insert-interval}, or they have different insertion intervals. If the insertion intervals are different, then by Lemma~\ref{lem:insert-between} the operations cannot affect each other, and so they have the same effect regardless of their order. So we need only analyze the case in which they have the same insertion interval $(k^\mathrm{before}, k^\mathrm{after})$.

If $o_r.\mathit{cur} = o_s.\mathit{cur}$, then by Lemma~\ref{lem:insert-conflict}, the operation with the greater operation ID appears first in the list, regardless of the order in which the operations are applied. If $o_r.\mathit{cur} \not= o_s.\mathit{cur}$, then one or both of the cursors must refer to a list element that appears between $k^\mathrm{before}$ and $k^\mathrm{after}$, and that did not yet exist in the common ancestor (Definition~\ref{def:common-ancestor}).

Take a cursor that differs from $k^\mathrm{before}$: the list element it refers to was inserted by a prior operation, whose cursor in turn refers to another prior operation, and so on. Following this chain of cursors for a finite number of steps leads to an operation $o_\mathrm{first}$ whose cursor refers to $k^\mathrm{before}$ (since an insertion operation always inserts at a position after the cursor).

Note that all of the operations in this chain are causally dependent on $o_\mathrm{first}$, and so they must have a Lamport timestamp greater than $o_\mathrm{first}$. Thus, we can apply the same argument as in Lemma~\ref{lem:insert-conflict}: if $\textsc{Insert}_2$ skips over the list element inserted by $o_\mathrm{first}$, it will also skip over all of the list elements that are causally dependent on it; if $\textsc{Insert}_1$ inserts a new element before $o_\mathrm{first}$, it is also inserted before the chain of operations that is based on it.

Therefore, the order of $o_r$ and $o_s$ in the final list is determined by the Lamport timestamps of the first insertions into the insertion interval after their common ancestor, in the chains of cursor references of the two operations. Since the argument above applies to all pairs of concurrent operations $o_r, o_s$ in $H$, we deduce that the final order of elements in the list depends only on the operation IDs but not the order of application, which shows that concurrent insertions to the same list are commutative.
\end{proof}

\begin{lemma}\label{lem:delete-commute}
In a history $H$, a deletion operation is commutative with respect to concurrent operations.
\end{lemma}

\begin{proof}
Given a deletion operation $o_d$ and any other concurrent operation $o_c$, we must show that the document state after applying both operations does not depend on the order in which $o_d$ and $o_c$ were applied.

The rules in Figure~\ref{fig:clear-rules} define how a deletion operation $o_d$ is applied: starting at the cursor in the operation, they recursively descend the subtree, removing $o_d.\mathit{deps}$ from the presence set $\mathsf{pres}(k)$ at all branch nodes in the subtree, and updating all registers to remove any values written by operations in $o_d.\mathit{deps}$.

If $o_c$ is an assignment or insertion operation, the \textsc{Assign} rule adds $o_c.\mathit{id}$ to the mapping from operation ID to value for a register, and the \textsc{Descend}, \textsc{Assign}, \textsc{Empty-Map} and \textsc{Empty-List} rules add $o_c.\mathit{id}$ to the presence sets $\mathsf{pres}(k)$ along the path through the document tree described by the cursor.

If $o_d.\mathit{cur}$ is not a prefix of $o_c.\mathit{cur}$, the operations affect disjoint subtrees of the document, and so they are trivially commutative. Any state changes by \textsc{Descend} and $\textsc{Add-ID}_1$ along the shared part of the cursor path are applied using the set union operator $\cup$, which is commutative.

Now consider the case where $o_d.\mathit{cur}$ is a prefix of $o_c.\mathit{cur}$. Since $o_c$ is concurrent with $o_d$, we know that $o_c.\mathit{id} \notin o_d.\mathit{deps}$. Therefore, if $o_c$ is applied before $o_d$ in the history, the \textsc{Clear-*} rules evaluating $o_d$ will leave any occurrences of $o_c.\mathit{id}$ in the document state undisturbed, while removing any occurrences of operations in $o_d.\mathit{deps}$.

If $o_d$ is applied before $o_c$, the effect on presence sets and registers is the same as if they had been applied in the reverse order. Moreover, $o_c$ applies in the same way as if $o_d$ had not been applied previously, because applying a deletion only modifies presence sets and registers, without actually removing map keys or list elements, and because the rules for applying an operation are not conditional on the previous content of presence sets and registers.

Thus, the effect of applying $o_c$ before $o_d$ is the same as applying $o_d$ before $o_c$, so the operations commute.
\end{proof}

\begin{lemma}\label{lem:assign-commute}
In a history $H$, an assignment operation is commutative with respect to concurrent operations.
\end{lemma}

\begin{proof}
Given an assignment $o_a$ and any other concurrent operation $o_c$, we must show that the document state after applying both operations does not depend on the order in which $o_a$ and $o_c$ were applied.

The rules \textsc{Assign}, \textsc{Empty-Map} and \textsc{Empty-List} define how an assignment operation $o_a$ is applied, depending on the value being assigned. All three rules first clear any causally prior state from the cursor at which the assignment is occurring; by Lemma~\ref{lem:delete-commute}, this clearing operation is commutative with concurrent operations, and leaves updates by concurrent operations untouched.

The rules also add $o_a.\mathit{id}$ to the presence set identified by the cursor, and \textsc{Descend} adds $o_a.\mathit{id}$ to the presence sets on the path from the root of the document tree described by the cursor. These state changes are applied using the set union operator $\cup$, which is commutative.

Finally, in the case where value assigned by $o_a$ is a primitive and the \textsc{Assign} rule applies, the mapping from operation ID to value is added to the register by the expression $\mathit{child}[\,\mathit{id} \mapsto \mathit{val}\,]$. If $o_c$ is not an assignment operation or if $o_a.\mathit{cursor} \not= o_c.\mathit{cursor}$, the operations are independent and thus trivially commutative.

If $o_a$ and $o_c$ are assignments to the same cursor, we use the commutativity of updates to a partial function: $\mathit{child}[\,\mathit{id}_1 \mapsto \mathit{val}_1\,]\,[\,\mathit{id}_2 \mapsto \mathit{val}_2\,] = \mathit{child}[\,\mathit{id}_2 \mapsto \mathit{val}_2\,]\,[\,\mathit{id}_1 \mapsto \mathit{val}_1\,]$ provided that $\mathit{id}_1 \not= \mathit{id}_2$. Since operation IDs (Lamport timestamps) are unique, two concurrent assignments add two different keys to the mapping, and their order is immaterial.

Thus, all parts of the process of applying $o_a$ have the same effect on the document state, regardless of whether $o_c$ is applied before or after $o_a$, so the operations commute.
\end{proof}

\begin{lemma}\label{lem:op-commute}
Given an operation history $H=\langle o_1, \dots, o_n \rangle$ from a valid execution, a new operation $o_{n+1}$ from that execution can be inserted at any point in $H$ after $o_{n+1}.\mathit{deps}$ have been applied. For all histories $H'$ that can be constructed this way, the document state resulting from applying the operations in $H'$ in order is the same, and independent of the ordering of any concurrent operations in $H$.
\end{lemma}

\begin{proof}
$H$ can be split into a prefix and a suffix, as described in the proof of Theorem~\ref{thm:convergence}. The suffix contains only operations that are concurrent with $o_{n+1}$, and we allow $o_{n+1}$ to be inserted at any point after the prefix. We then prove the lemma case-by-case, depending on the type of mutation in $o_{n+1}$. 

If $o_{n+1}$ is a deletion, by Lemma~\ref{lem:delete-commute} it is commutative with all operations in the suffix, and so $o_{n+1}$ can be inserted at any point within, before, or after the suffix without changing its effect on the final document state. Similarly, if $o_{n+1}$ is an assignment, by Lemma~\ref{lem:assign-commute} it is commutative with all operations in the suffix.

If $o_{n+1}$ is an insertion, let $o_c$ be any operation in the suffix, and consider the cases of $o_{n+1}$ being inserted before and after $o_c$ in the history. If $o_c$ is a deletion or assignment, it is commutative with $o_{n+1}$ by Lemma~\ref{lem:delete-commute} or Lemma~\ref{lem:assign-commute} respectively. If $o_c$ is an insertion into the same list as $o_{n+1}$, then by Lemma~\ref{lem:insert-reorder} the operations are commutative. If $o_c$ is an insertion into a different list in the document, its effect is independent from $o_{n+1}$ and so the two operations can be applied in any order.

Thus, $o_{n+1}$ is commutative with respect to any concurrent operation in $H$. Therefore, $o_{n+1}$ can be inserted into $H$ at any point after its causal dependencies, and the effect on the final document state is independent of the position at which the operation is inserted.
\end{proof}

This completes the induction step in the proof of Theorem~\ref{thm:convergence}, and thus proves convergence of our datatype.

\fi 
\end{document}